\documentclass[11pt]{article}
\usepackage{verbatim,url,enumerate,color,paralist}
\usepackage{amsmath,amsfonts}
\usepackage{epsfig,amssymb,amstext,xspace,theorem}
\usepackage{algorithm}
\usepackage{algorithmicx,algpseudocode}
\usepackage{fullpage}
\usepackage{cite,wrapfig}

\newcommand{\qed}{\hspace*{\fill}$\Box$}
\newtheorem{theorem}{Theorem}[section]
\newtheorem{lemma}[theorem]{Lemma}
\newtheorem{corollary}[theorem]{Corollary}
\newtheorem{definition}[theorem]{Definition}

\newtheorem{claim}[theorem]{Claim}

\allowdisplaybreaks

\newenvironment{proof}[1][Proof. ]{\noindent {\bf #1 }}{\qed}

\newcommand{\argmin}{{\rm argmin}}

\newcommand{\Xomit}[1]{}

\begin{document}

\title{Submodular Approximation: 
Sampling-based Algorithms and Lower Bounds\thanks{This work supported in part 
by NSF grant CCF-0728869. A preliminary version of this paper has appeared in the Proceedings of the 49th Annual IEEE Symposium on Foundations of Computer Science.
} }

\author{\vspace{2mm}{Zoya Svitkina\thanks{Department of Computing Science, University of Alberta, Canada. 
} \hspace{2cm} Lisa Fleischer\thanks{Department of Computer Science, Dartmouth, USA.
}}
}

\date{May 31, 2010}

\maketitle

\begin{abstract}
We introduce several generalizations of classical computer science problems 
obtained by  
replacing  simpler objective functions with general submodular functions.
The new problems include
submodular load balancing, which generalizes load balancing or minimum-makespan 
scheduling, submodular sparsest cut and submodular
balanced cut, which generalize their respective graph cut problems, as well as 
submodular function minimization with a cardinality lower bound. 
We establish upper and lower bounds for the approximability of these problems 
with a polynomial number of queries to a function-value oracle.
The approximation guarantees for most of our algorithms are of the order of
$\sqrt{{n}/{\ln n}}$. We show that this is the inherent difficulty of the problems
by proving matching  lower bounds.

We also give an improved lower bound for the problem of approximating 
a monotone submodular function everywhere.  In addition, we present an algorithm for 
approximating submodular functions with special structure, whose 
guarantee is close to the lower bound. Although quite restrictive, the class 
of functions with this structure  
includes the ones that are used for lower bounds both by us and in 
previous work.  This demonstrates that if there are significantly stronger lower 
bounds for this problem, they rely on more general submodular functions.
\end{abstract}

\section{Introduction}
A function $f$ defined on subsets of a ground set $V$ is called \emph{submodular} if  
for all subsets $S, T\subseteq V$,  $f(S) + f(T) \geq f(S\cup T) + f(S\cap T)$.
Submodularity is a discrete analog of convexity.
It also shares some nice properties with concave functions, as it captures
decreasing marginal returns.  Submodular functions generalize
cut functions of graphs and rank functions of matrices and matroids,
and arise in a variety of applications including facility location,
assignment, scheduling, and network design.

In this paper, we introduce and study several generalizations of classical 
computer science problems. These new problems have a general submodular 
function in their objectives, in place of much simpler functions in the 
objectives of their classical counterparts. 
The problems include 
\emph{submodular load balancing}, which generalizes 
load balancing or minimum-makespan scheduling, and \emph{submodular 
minimization with cardinality lower bound}, which generalizes the minimum 
knapsack problem. In these two  
problems, the size of a collection of items, instead of being just a sum of 
their individual sizes, is now a submodular function. Two other
new problems are \emph{submodular sparsest cut} and \emph{submodular 
balanced cut}, which generalize their respective graph cut problems. Here, a 
general submodular function replaces the graph cut function, which itself is a 
well-known special case of a submodular function.  The last problem that we study is \emph{approximating a submodular function everywhere}.
All of these problems are defined on a set $V$ of $n$ elements with
a nonnegative submodular function $f:2^V\rightarrow \mathbb{R}_{\geq 0}$. 
Since the amount of information necessary to convey a general submodular
function may be exponential in $n$,
we rely on value-oracle access to $f$ to develop algorithms with running time 
polynomial in $n$. A \emph{value oracle} for $f$ is
a black box that, given a subset $S$, returns the value $f(S)$.
%
The following are formal definitions of the  problems.

\textbf{Submodular Sparsest Cut (SSC):}
Given a set of unordered pairs 
$\{\{u_i, v_i\}~|~ u_i,v_i\in V\}$, 
each with a demand $d_i>0$, find a subset $S\subseteq V$ minimizing
${f(S)}/{\sum_{i:|S\cap \{u_i,v_i\}|=1}d_i}$.
The denominator is the amount of demand separated by the ``cut'' 
$(S,\bar{S})$\footnote{For any set $S\subseteq V$, we use  
$\bar{S}$ to denote its complement set, $V\setminus S$.}.
In \emph{uniform} SSC, all pairs of nodes have demand equal to 
one, so the objective function is ${f(S)}/{|S||\bar{S}|}$.
Another special case is the 
\emph{weighted} SSC problem, in which each  element $v\in V$ has a non-negative 
weight $w(v)$,  and the demand between any pair of 
elements $\{u,v\}$ is equal to the product $w(u)\cdot w(v)$. 

\textbf{Submodular $b$-Balanced Cut (SBC):}
Given a weight function 
$w:V\rightarrow \mathbb{R}_{\geq 0}$, a cut $(S, \bar{S})$ is called 
$b$-balanced (for $b\leq \frac{1}{2}$) if $w(S)\geq b\cdot w(V)$ and 
$w(\bar{S})\geq b\cdot w(V)$, where $w(S)=\sum_{v\in S}w(v)$. 
The goal of the problem is to find a $b$-balanced cut $(S, \bar{S} )$ that 
minimizes $f(S)$. In the \emph{unweighted} special case, the weights of all elements are equal to one.

\textbf{Submodular Minimization with Cardinality Lower Bound (SML):}
For a given $W\geq 0$, find a subset $S\subseteq V$ with $|S|\geq W$ that minimizes 
$f(S)$. A generalization with 0-1 weights $w:V\rightarrow \{0,1\}$ is to 
find $S$ with $w(S) \geq W$ minimizing $f(S)$.  

\textbf{Submodular Load Balancing (SLB):}
The \emph{uniform} version is to find, given a monotone\footnote{ A function 
$f$ is monotone if $f(S) \leq f(T)$ whenever $S\subseteq T$.} submodular 
function $f$ and a positive integer $m$, a partition of $V$ into  
$m$ sets, $V_1,\dots, V_m$ (some possibly empty), so as to 
minimize $\max_i f(V_i)$.  
The \emph{non-uniform} version is to find, for $m$ monotone submodular functions
$f_1, \ldots, f_m$ on $V$, a partition 
$V_1,\dots, V_m$ that minimizes $\max_i f_i(V_i)$.  

\textbf{Approximating a Submodular Function Everywhere:}
Produce a function $\hat{f}$ (not necessarily submodular) that for all 
$S\subseteq V$ satisfies $\hat{f}(S) \leq f(S) \leq \gamma(n) \hat{f}(S)$,
with approximation ratio $\gamma(n) \geq 1$ as small as possible. 
We also consider the special case of 
monotone two-partition functions, which we define as follows. 
A submodular function $f$ on a ground set $V$ 
is a \emph{two-partition} (2P) function if there is a set $R\subseteq V$ 
such that for all sets $S$, the value of $f(S)$ depends only on the sizes 
$|S\cap R|$ and $|S \cap \bar{R}|$.

\subsection{Motivation}
Submodular functions arise in a variety of contexts, often in
optimization settings. 
The problems that we define in this paper use submodular functions to 
generalize some of the best-studied problems in computer science.  
These generalizations capture many variants of their corresponding classical 
problems. For example, the submodular sparsest and balanced cut problems 
generalize not only graph cuts, but also hypergraph cuts. In addition, they may 
be useful as subroutines for solving other problems, in the same way that  
sparsest and balanced cuts are used for approximating graph problems, such as 
the minimum cut linear arrangement, often as part of divide-and-conquer 
schemes.   
  The SML problem can model a 
scenario in which  costs  follow economies of scale, and a certain number 
of items has to be bought at the minimum total cost. An example application 
of SLB is compressing and storing files on multiple hard drives or servers 
in a load-balanced way. Here the size of a compressed collection of files may 
be much smaller than the sum of individual file sizes, and modeling it by a 
monotone submodular function is reasonable considering that the entropy 
function is known to be monotone and submodular \cite{fujishige:polymatroid}.


\subsection{Related work}
Because of the relation of submodularity to cut functions and
matroid rank functions, and their exhibition of 
decreasing marginal returns, there has been substantial
interest in optimization problems involving submodular functions.
Finding the set that has the minimum function value is
a well-studied problem that was first shown to be polynomially
solvable using the ellipsoid method~\cite{GLS81,GLS88}.  Further
research has yielded several more combinatorial 
approaches~\cite{FleischerIwata03,Iwata03,Iwata08,iwata:fleischer:fujishige,iwata:orlin,Orlin07,Queyranne98,schrijver:combinatorial}.

Submodular functions arise in
facility location and assignment problems, and this has
spawned interest in the problem of finding the set with
the maximum function value.  Since this is NP-hard,
research has focused on approximation algorithms
for maximizing monotone or non-monotone submodular functions, perhaps subject to cardinality or other constraints~\cite{calinescu:submod,feige-maxsubmod,kulik,lee:nonmon,lee:sviridenko:vondrak,nemhauser:wolsey:fisher,sviridenko:note}.
A general approach for deriving inapproximability results for such maximization problems is presented in \cite{vondrak:symmetry}.

Research on other optimization problems that involve submodular functions includes \cite{calinescu:zelikovsky,chekuri:pal:recursive,hayrapetyan:swamy:tardos,svitkina:tardos:facility,SwamySW07,wolsey:analysis}. 
Zhao et al.\ \cite{ZhaoNagamochiIbaraki} study a submodular multiway 
partition problem, which is similar to our SLB problem, except that the subsets are required to be non-empty and the objective is the sum of function values on the  subsets, as opposed to the maximum. 
Subsequent to the publication of the preliminary version of this paper, generalizations of other  combinatorial problems to submodular costs have been defined, with upper and lower bounds derived for them. These include the set cover problem and its special cases vertex cover and edge cover, studied in \cite{iwata:nagano}, as well as vertex cover, shortest path, perfect matching, and spanning tree studied in \cite{goel:karande}. In \cite{goel:karande}, extensions to the case of multiple agents (with different cost functions) are also considered.

Since it is impossible to learn a general 
submodular function exactly without looking at the function
value on all (exponentially many) subsets~\cite{Cunningham85}, 
there has been recent interest in approximating submodular functions everywhere with a polynomial number of value oracle queries.  
Goemans et al.\ \cite{goemans:learning} give an algorithm that approximates an arbitrary monotone submodular function to a factor $\gamma(n)=O(\sqrt{n} \log n)$, and approximates a rank function of a matroid to a factor $\gamma(n)=\sqrt{n+1}$. 
A lower bound of $\Omega \left(\frac{\sqrt{n}}{\ln n}\right)$ for this problem on monotone functions and an improved lower bound of 
$\Omega \left(\sqrt{\frac{n}{\ln n}}\right)$ for 
non-monotone functions were obtained 
in \cite{goemans:harvey:kleinberg:mirrokni,goemans:learning}. These lower bounds apply to all algorithms that make a polynomial number of value-oracle queries.

All of the optimization problems that we consider in this
paper are known to be NP-hard even when the objective
function can be expressed compactly as a linear or
graph-cut function.  While there is an FPTAS for the
minimum knapsack problem~\cite{gens:levner}, the
best approximation for load balancing
on uniform machines is a PTAS~\cite{HochbaumShmoys87}, and
on unrelated machines the best possible upper and lower bounds are 
constants~\cite{LenstraST90}. The best approximation known
for the sparsest cut problem is 
$O(\sqrt{\log n})$~\cite{AroraHK04,arora:rao:vazirani:sqrt-logn}, and the 
balanced cut problem is approximable to a factor of $O(\log n)$ 
\cite{racke:optimal}. 
For the special case of SML on graphs, introduced in 
\cite{svitkina:tardos:multiway-cuts}, an $O(\log n)$ approximation is possible 
using the recent results of  R\"acke \cite{racke:optimal}.

\subsection{Our results and techniques}
We establish upper and lower bounds for the approximability of the problems 
listed above. Surprisingly, these factors are 
quite high. Whereas the corresponding classical problems are approximable to 
constant or logarithmic factors, the guarantees that we prove for most 
of our algorithms are of the order of $\sqrt{\frac{n}{\ln n}}$. We  show that 
this is the inherent difficulty of these problems by proving matching (or, in 
some cases, almost matching) lower 
bounds. Our lower bounds are unconditional, and rely on the 
difficulty of distinguishing different submodular functions by performing only 
a polynomial number of queries in the oracle model. The proofs are based on 
the techniques in \cite{feige-maxsubmod,goemans:learning}. To prove the upper bounds, we present randomized approximation algorithms which use their randomness for 
sampling subsets of the ground set of elements. We show that with relatively 
high probability (inverse polynomial), a sample can be obtained such that its 
overlap with the optimal set is significantly higher than expected. 
Using the samples, the algorithms employ submodular function minimization to 
find   
candidate solutions. This is done in such a way that if the sample does indeed 
have a large overlap with the optimal set, then the solution 
satisfies the algorithm's guarantee.

For SSC and uniform SLB, we show that they can be approximated to a 
$\Theta \left(\sqrt{\frac{n}{\ln n}} \right)$ factor. For SBC, we use the 
weighted SSC as a subroutine, which allows us  to obtain a bicriteria 
approximation in a similar way as Leighton and Rao \cite{leighton:rao} do for 
graphs. For SML, we also consider bicriteria results. For $\rho \geq 1$ and 
$0<\sigma \leq 1$, a $(\rho,\sigma)$-approximation for SML is an algorithm that 
outputs a set $S$ such that $f(S)\leq \rho B $ and $w(S)\geq \sigma W$, 
whenever the input instance contains a set $U$ with $f(U)\leq B$ and 
$w(U)\geq W$. We present a lower bound showing that there is no   
$(\rho,\sigma)$ approximation for any $\rho$ and $\sigma$ with 
$\frac{\rho}{\sigma} = o\left(\sqrt{\frac{n}{\ln n}}\right)$.  For 0-1 
weights, we  obtain a $\left(5 \sqrt{\frac{n}{\ln n}}, \frac{1}{2}\right)$  
approximation. This  algorithm can be used to obtain an $O(\sqrt{n \ln n})$ 
approximation for non-uniform SLB. 

We briefly note here that one can consider the problem of minimizing a 
submodular function with an \emph{upper} bound on cardinality (i.e., 
minimize $f(S)$ subject to $|S|\leq W$). 
For this problem, 
a $(\frac{1}{\alpha},\frac{1}{1-\alpha})$ bicriteria approximation is possible  
for any $0 < \alpha < 1$, using techniques in 
\cite{hayrapetyan:kempe:pal:svitkina}. For non-bicriteria algorithms, a hardness result of $\Omega\left(\sqrt{\frac{n}{\ln n}}\right)$ follows by reduction from SML, using the submodular function $\bar{f}$, defined as $\bar{f}(S) = f(\bar{S})$, and a cardinality bound $\overline{W} = n-W$.

For  approximating \emph{monotone} submodular functions everywhere, our lower bound is $\Omega\left(\sqrt{\frac{n}{\ln n}}\right)$, 
which improves the bound for monotone functions 
in~\cite{goemans:harvey:kleinberg:mirrokni,goemans:learning}, 
and matches the lower bound for arbitrary submodular functions, also 
in~\cite{goemans:harvey:kleinberg:mirrokni,goemans:learning}.
Our lower bound proof for this problem, as well as the earlier ones, 
use 2P functions, and thus still hold for this special case. We  show that monotone 2P functions can be approximated within a factor $O(\sqrt{n})$.   
Besides leaving a relatively small gap between the upper and lower bounds, 
this shows that if much stronger lower bounds for the approximation problem exist, 
they rely on more general submodular functions. 

For the problems studied in this paper, our lower bounds  show 
the impossibility of constant or even polylogarithmic approximations  in the 
value oracle  model. This means that in order to obtain better results for 
specific applications, one has to resort to more restricted models, 
avoiding the full generality of arbitrary submodular functions.

\section{Preliminaries}
In the analysis of our algorithms, we repeatedly use the facts that the sum of  
submodular  functions is submodular, and that submodular functions can be 
minimized in polynomial time.  
For example, this allows us to minimize (over $T\subseteq V$) expressions like 
$f(T)-\alpha\cdot|T\cap S|$, where $\alpha$ is a constant and $S$ is a fixed 
subset of $V$.


We present our algorithms by providing 
a \emph{randomized relaxed decision procedure} for each of the problems.  Given
an instance of a minimization problem, a target value $B$, and a probability 
$p$, this procedure either declares that the problem is infeasible (outputs \emph{fail}),
or finds a solution to the instance with objective value
at most $\gamma B$, where $\gamma$ is the approximation factor. 
We say that an instance is feasible if it has a solution with cost strictly less than 
$B$ (we use strict inequality for technical reasons; this can be avoided by 
adding a small value $\varepsilon>0$ to $B$). 
The guarantee provided with each decision procedure is that for any feasible 
instance, it outputs a $\gamma$-approximate solution with probability at 
least $p$. On an infeasible instance, either of the two outcomes is allowed.
Randomized relaxed decision procedures can be turned into randomized 
approximation  
algorithms by finding upper and lower bounds for the optimum and performing 
binary search. Our algorithms run in time polynomial in $n$ and 
$\ln\frac{1}{1-p}$. 

Let us say that an algorithm \emph{distinguishes} two functions $f_1$ and $f_2$ 
if it produces different output if given (an oracle for) $f_1$ as input than 
if given (an oracle for) $f_2$. The following result is used for obtaining 
all of our lower bounds.

\begin{lemma}\label{lem:alg}
Let $f_1$ and $f_2$ be two set functions, with $f_2$, but not $f_1$, 
parametrized by a string of random bits $r$. If for any set $S$, chosen 
without knowledge of $r$, the probability (over $r$) that $f_1(S)\neq f_2(S)$ is 
$n^{-\omega(1)}$, then 
any algorithm that makes a polynomial number of oracle queries has 
probability at most $n^{-\omega(1)}$ of distinguishing $f_1$ and $f_2$.
\end{lemma}
\begin{proof}
We use reasoning similar to \cite{feige-maxsubmod}. Consider first a 
deterministic algorithm and the computation path that it follows if it 
receives the values of $f_1$ as answers to all its oracle queries. 
Note that this is a single computation path that does not depend on $r$, 
because $f_1$ does not depend on $r$. On this path the algorithm makes some
polynomial number of oracle queries, say $n^a$. 
Using the union bound, we 
know that the probability that $f_1$ and $f_2$ differ on any of these $n^a$ sets is 
at most $n^a\cdot n^{-\omega(1)}=n^{-\omega(1)}$. 
So, with probability at least $1-n^{-\omega(1)}$, if given either $f_1$ or 
$f_2$ as input, the algorithm only queries sets for which $f_1=f_2$, and 
therefore stays on the same computation path, producing the same  answer in
both cases.

A randomized algorithm can be viewed as a distribution over a set of 
deterministic algorithms.  Since, by the discussion above, each of these 
deterministic algorithms has probability at most $n^{-\omega(1)}$ of 
distinguishing $f_1$ and $f_2$, the randomized algorithm as a whole also has 
probability at most $n^{-\omega(1)}$ of distinguishing these two functions. 
\end{proof}

\medskip

The following theorem about random sampling is used for bounding 
probabilities in the analyses of our algorithms. We use the constant 
$c=1/(4\sqrt{2\pi})$ throughout the paper.

\begin{theorem} \label{thm:gen}
Suppose that $m$ elements are selected independently, with probability $0<q<1$ 
each. Then for 
$0 \leq \varepsilon < \frac{1-q}{q}$, 
the probability that exactly $\lceil qm(1+\varepsilon) \rceil$ elements are selected is at least 
${cq}\cdot {m^{-\frac{3}{2}}}\cdot \exp\left[\frac{-\varepsilon^2qm}{1-q}\right]$. 
\end{theorem}
\begin{proof}
Let $\lambda = qm(1+\varepsilon)$. 
First we consider the case that $\lambda$ is integer. 
For convenience, let $\kappa = q(1+\varepsilon)$, and note that $\kappa<1$.
Using an approximation that 
$\sqrt{2\pi n} \left(\frac{n}{e}\right)^n \leq n! \leq 
2 \sqrt{2\pi n} \left(\frac{n}{e}\right)^n$,  
which is derived from Stirling's formula \cite[p.\ 55]{clrs}, we obtain the bound 
\begin{eqnarray*}
{m \choose m\kappa} ~=~ \frac{m!}{(m\kappa)! (m-m\kappa)!} &\geq&
\frac{\sqrt{2\pi}}{(2\sqrt{2\pi})^2} \cdot
\frac{\sqrt{m}}{\sqrt{m\kappa}\sqrt{m-m\kappa}}\cdot
\frac{  (m/e)^m}{  (m\kappa/e)^{m\kappa}  ((m-m\kappa)/e)^{m-m\kappa}}\\
&\geq& \frac{1}{4\sqrt{2\pi}} \cdot\frac{1}{\sqrt{m}} \cdot
\frac{1}{\kappa^{m\kappa} (1-\kappa)^{m-m\kappa}}.
\end{eqnarray*}
Let $X$ be the number of elements selected in the random experiment. 
Then
\begin{eqnarray*}
\Pr[X=m\kappa] 
~=~
{m \choose {m\kappa}} q^{m \kappa} (1-q)^{m-m\kappa}
&\geq& 
\frac{c}{\sqrt{m}} \cdot \frac{q^{m \kappa}\cdot (1-q)^{m-m \kappa}}
{\kappa^{m \kappa} \cdot \left(1-\kappa\right)^{m-m\kappa}}\\ 
&=&
\frac{c}{\sqrt{m}} \cdot 
\left(\frac{1}{1+\varepsilon}\right)^{m \kappa}\cdot
\left(\frac{1-q}{1-q(1+\varepsilon)}\right)^{m-m\kappa} \\
&=&
\frac{c}{\sqrt{m}} \cdot 
\frac{1}{(1+\varepsilon)^{m\kappa}}\cdot
\frac{1}{\left(1-\frac{\varepsilon q}{1-q}\right)^{m-m\kappa}}\\ 
&\geq&
\frac{c}{\sqrt{m}} \cdot 
\exp\left[{-\varepsilon m \kappa + \frac{\varepsilon q}{1-q}m(1-\kappa)}\right],
\end{eqnarray*}
where we have used the inequality that $1+x\leq e^x$ for all $x$. The 
assumption that $\varepsilon < \frac{1-q}{q}$ ensures that the  
denominator ${1-q(1+\varepsilon)}$ is positive.
Now, the exponent of $e$ is equal to 
$${-\varepsilon qm(1+\varepsilon) + \frac{\varepsilon q}{1-q}m(1-q-\varepsilon q)} ~=~ 
-\varepsilon qm - \varepsilon^2 qm + \varepsilon qm - \frac{\varepsilon^2q^2m}{1-q}
~=~  \frac{-\varepsilon^2 qm}{1-q}.$$
Noting that $c\cdot m^{-\frac{1}{2}} \geq cq\cdot m^{-\frac{3}{2}}$ concludes the proof for the case that $\lambda$ is integer.

If $\lambda$ is fractional, then $\lceil \lambda \rceil = \lfloor \lambda \rfloor + 1$. Then
\begin{equation} \label{eq:rat}
\frac{\Pr[X=\lceil \lambda \rceil]}{\Pr[X=\lfloor \lambda \rfloor]}
~=~ \frac{{m \choose \lfloor \lambda \rfloor+1} \ q^{\lfloor \lambda \rfloor+1} \ (1-q)^{m-\lfloor \lambda \rfloor-1}}{{m \choose \lfloor \lambda \rfloor} \ 
q^{\lfloor \lambda \rfloor} \  (1-q)^{m-\lfloor \lambda \rfloor}}
~=~ \frac{(m-\lfloor \lambda \rfloor) \ q}{(\lfloor \lambda \rfloor+1) \ (1-q)}.
\end{equation}
As $\varepsilon\geq 0$, we have $\lambda\geq qm$. Now consider the case that 
$\lfloor \lambda \rfloor \leq qm$. As $qm$ is the expectation of $X$, either
$\lceil \lambda \rceil$ or $\lfloor \lambda \rfloor$ is the most likely value of $X$, having probability of at least $\frac{1}{m+1}$. In the first case, 
$\Pr[X=\lceil \lambda \rceil]\geq \frac{1}{m+1} \geq \frac{c}{m}$, and we are done. In the second case, using sequentially (\ref{eq:rat}), 
$\lfloor \lambda \rfloor \leq qm$, and 
$\lfloor \lambda \rfloor +1 = \lceil \lambda \rceil \leq m$ 
(which is implied by $\kappa <1$ above), we obtain the result:
$$\Pr[X=\lceil \lambda \rceil] 
~\geq~ \frac{1}{m+1} \cdot 
\frac{(m-\lfloor \lambda \rfloor) \ q}{(\lfloor \lambda \rfloor+1) \ (1-q)}
~\geq~ \frac{1}{m+1} \cdot \frac{mq}{\lfloor \lambda \rfloor+1}
~\geq~ \frac{cq}{m}.
$$

The remaining case is that $\lfloor \lambda \rfloor > qm$. Define $\varepsilon'>0$ to be such that $qm(1+\varepsilon') = \lfloor qm(1+\varepsilon) \rfloor = \lfloor \lambda \rfloor$. Note that $\varepsilon' \leq \varepsilon$. 
Applying the proof that we used for integer $\lambda$, we obtain that 
$$\Pr[X=\lfloor \lambda \rfloor] 
~\geq~ \frac{c}{\sqrt{m}} \cdot \exp\left[\frac{-\varepsilon'^2qm}{1-q}\right] 
~\geq~ \frac{c}{\sqrt{m}} \cdot \exp\left[\frac{-\varepsilon^2qm}{1-q}\right],
$$ 
where we also used monotonicity of the exponential function. Using the fact that 
$\lfloor \lambda \rfloor \leq m-1$, we simplify equation (\ref{eq:rat}) to obtain that ${\Pr[X=\lceil \lambda \rceil]}/{\Pr[X=\lfloor \lambda \rfloor]} \geq \frac{q}{m}$. Together with the above inequality, this gives the desired result.
\end{proof}


\section{Submodular sparsest cut and submodular balanced cut}
\label{sec:ssc}

\subsection{Lower bounds}
Let $\varepsilon>0$ be such that 
$\varepsilon^2 = \frac{1}{n}\cdot \omega(\ln n)$, let  
$\beta=\frac{n}{4}(1+\varepsilon)$, and let $R$ be a subset of $V$ of size $\frac{n}{2}$, with parameters such that $n$ is even and $\beta$ is an integer. We define the following two functions, and show that they are submodular and hard to distinguish. Moreover, these functions are symmetric\footnote{A function $f$ is symmetric if $f(S)=f(\bar{S})$ for all $S$.}.
\begin{eqnarray*}
f_1(S) &=& \min\left(|S|,~\frac{n}{2}\right)-\frac{|S|}{2}\\
f_2(S) &=& \min\left(|S|,~\frac{n}{2},~ \beta+|S\cap R|,~ \beta+|S\cap \bar{R}| \right) -\frac{|S|}{2}
\end{eqnarray*}

\begin{lemma}
Functions $f_1$ and $f_2$ defined above are nonnegative, submodular, and symmetric.
\end{lemma}
\begin{proof}
The first function can be written as $f_1(S)=\frac{1}{2} \min(|S|, |\bar{S}|)$, which makes it easy to see that it is nonnegative and symmetric. 
It suffices to show that $f(S)=\min(|S|,\frac{n}{2})$ is submodular, since $-\frac{|S|}{2}$ is modular\footnote{A modular function is one for which the submodular inequality is satisfied with equality.}. We use an alternative definition of submodularity: $f$ is submodular if 
for all $S\subset V$ and $a,b\in V\setminus S$, with $a\neq b$, it holds that 
$f(S\cup\{a,b\}) - f(S\cup \{b\}) \leq f(S\cup\{a\}) - f(S)$.
The only way that this inequality can be violated for our function is if $f(S\cup\{a,b\}) - f(S\cup \{b\})=1$ and $f(S\cup\{a\}) - f(S)=0$. But this is a contradiction, since the second part implies that $|S|\geq n/2$, and the first one implies that $|S\cup \{b\}|<n/2$.

To see that $f_2(S)$ is nonnegative, we note that $\beta +|S\cap R|-\frac{|S|}{2} \geq \frac{n}{4} +|S\cap R| - \frac{|S\cap R|}{2} - \frac{|S\cap \bar{R}|}{2} \geq 0$, since $|S\cap \bar{R}|\leq \frac{n}{2}$.
A similar calculation shows that $\beta +|S\cap \bar{R}|-\frac{|S|}{2}\geq 0$, and thus $f_2(S)\geq 0$ for all $S$. 
To show symmetry, we use the fact that $|R|=\frac{n}{2}$, and thus 
$$|S\cap R|-\frac{|S|}{2} = \frac{n}{2} - |\bar{S}\cap {R}| -\frac{|S|}{2} = \frac{|\bar{S}|}{2}- |\bar{S}\cap {R}| = -\frac{|\bar{S}|}{2} + |\bar{S}| - |\bar{S}\cap {R}| = |\bar{S}\cap\bar{R}| - \frac{|\bar{S}|}{2}.$$ 
Analogously, 
$|S\cap \bar{R}|-\frac{|S|}{2} = |\bar{S}\cap {R}|- \frac{|\bar{S}|}{2}$. 
Thus, we have that 
\begin{eqnarray*}
f_2(S) &=& \min\left( \frac{|S|}{2},~ \frac{|\bar{S}|}{2},~ 
\beta + |S\cap R|-\frac{|S|}{2},~ 
\beta +|S\cap \bar{R}|-\frac{|S|}{2} \right) \\
&=& \min\left( \frac{|S|}{2},~ \frac{|\bar{S}|}{2},~ 
\beta + |\bar{S}\cap\bar{R}| - \frac{|\bar{S}|}{2},~ 
\beta + |\bar{S}\cap {R}|- \frac{|\bar{S}|}{2} \right) ~=~ f_2(\bar{S}).
\end{eqnarray*}

For submodularity of $f_2$, we focus only on $f(S)=\min\left(|S|, \frac{n}{2}, \beta+|S\cap R|, \beta+|S\cap \bar{R}| \right)$.  Suppose for the sake of contradiction that for some $a,b \in V$, we have $f(S\cup \{a,b\})-f(S\cup \{b\})=1$ but $f(S\cup \{a\})-f(S)=0$. We assume that $a\in R$ (the case that $a\in \bar{R}$ is similar).
First consider the case that $b$ is also in the set $R$. In this case $f(S\cup \{a\})=f(S\cup \{b\})$. The fact that the function value does not increase when $a\in R$ is added to $S$ means that the minimum is achieved by one of the terms that do not depend on $|S\cap R|$, namely $f(S)=\min(\frac{n}{2}, \beta+|S\cap \bar{R}|)$. But then the minimum would also not increase when the second element of $R$ is added, and we would have $f(S\cup \{a,b\})=f(S\cup \{b\})$, contradicting the assumption.

The remaining case is that $a\in R$ and $b\in \bar{R}$. As before, 
$f(S)=\min(\frac{n}{2}, \beta +|S\cap \bar{R}|)$. But if $f(S)=\frac{n}{2}$, then 
$f(S\cup \{a,b\})=\frac{n}{2}$, which contradicts our assumptions. So 
$f(S)=\beta +|S\cap \bar{R}|$. Now, $f(S\cup \{b\})$ increases from the addition of $a\in R$, which means that its minimum is achieved by a term that depends 
on $|S\cap R|$: 
$f(S\cup \{b\})=\min(|S|+1, \beta + |S\cap R|)$. Suppose that $f(S\cup \{b\})=|S|+1$. This means that $|S|+1 \leq \beta + |(S\cup \{b\})\cap \bar{R}| = \beta +|S\cap \bar{R}|+1$. But we also know that 
$\beta + |S\cap \bar{R}|\leq |S|$ (from the fact that 
$f(S)=\beta +|S\cap \bar{R}|$). Thus, $|S|=\beta + |S\cap \bar{R}|$ and 
$f(S\cup \{b\})=\beta + |S\cap \bar{R}|+1=\beta + |(S\cup \{b\})\cap \bar{R}|$. But this term does not depend on $|S\cap R|$, so adding $a\in R$ to $S\cup \{b\}$ would not change the function value, a contradiction. Finally, suppose that $f(S\cup \{b\})=\beta + |S\cap R|$. As $f(S)=\beta+|S\cap \bar{R}|$, we know that $\beta+|S\cap \bar{R}|\leq |S|$, and therefore $\beta \leq |S\cap R|$. So $f(S\cup \{b\})=\beta + |S\cap R| \geq 2\beta > \frac{n}{2}$, by the definition of $\beta$. But this is a contradiction, as the value of $f$ is always at most $\frac{n}{2}$.
\end{proof}

\medskip

To give a lower bound for SSC and SBC, we prove the following result and then apply Lemma \ref{lem:alg} to show that the functions $f_1$ and $f_2$ above are hard to distinguish.

\begin{lemma}\label{lem:disting}
Fix an arbitrary subset $S\subseteq V$, and then let $R$ be a random subset of $V$ of size $\frac{n}{2}$. Then the probability (over the choice of $R$) that $f_1(S)\neq f_2(S)$ is at most $n^{-\omega(1)}$.
\end{lemma}
\begin{proof}
We note that $f_1(S)\neq f_2(S)$ if and only if 
$\min(\beta+|S\cap R|, \beta+|S\cap \bar{R}|) < \min(|S|,\frac{n}{2})$.
This happens if either $\beta+|S\cap {R}| < \min(|S|,\frac{n}{2})$ or $\beta+|S\cap \bar{R}| < \min(|S|,\frac{n}{2})$. The probabilities of these two events are equal, so let us denote one of them by $p(S)$. If 
we show that $p(S) = n^{-\omega(1)}$, then the lemma follows by an application of the union bound. 

First, we claim that $p(S)$ is maximized when 
$|S|=\frac{n}{2}$. For this, suppose that $|S|\geq \frac{n}{2}$. Then 
$p(S)=\Pr[\beta+|S\cap {R}| < \frac{n}{2}]$. But this probability can only increase if an element is removed from $S$. 
Similarly, in the case that $|S|\leq \frac{n}{2}$,
$p(S)=\Pr[\beta+|S\cap {R}| < |S|] = \Pr[\beta < |S\cap \bar{R}|]$. But this probability can only increase if an element is added to $S$.

For a set $S$ of size $\frac{n}{2}$, 
$p(S)=\Pr[\beta+|S\cap {R}| < \frac{n}{2}] = \Pr[|S\cap {R}| < \frac{n}{4} (1-\varepsilon)]$.
If instead of choosing $R$ as a random subset of $V$ of size $\frac{n}{2}$, we consider a set $R'$ for which each element is chosen independently with probability $\frac{1}{2}$, then $p(S)$ becomes 
\begin{eqnarray*}
p(S) &=& \Pr\left[|S\cap 
R'| < \frac{n}{4} (1-\varepsilon) ~\left|~ |R'|=\frac{n}{2} \right.\right] \\
&=& \frac{\Pr\left[|S\cap R'| < \frac{n}{4} (1-\varepsilon) \wedge |R'|=\frac{n}{2} \right]}
{\Pr\left[|R'|=\frac{n}{2}\right] }\\
&\leq& (n+1) \cdot \Pr\left[|S\cap R'| < \frac{n}{4} (1-\varepsilon)\right]. 
\end{eqnarray*}
This allows us to make a switch to independent variables, so that we can use Chernoff bounds \cite{motwani:raghavan}.  
The expectation $\mu$ of 
$|S\cap R'|$ is equal to  $|S|/2 = n/4$, so 
$$\Pr\left[|S\cap R'|<(1-\varepsilon)\mu\right] ~<~ 
e^{-\mu \varepsilon^2/2} ~=~ e^{-\omega(\ln n)} ~=~ 
{n^{-\omega(1)}},$$
remembering that $\varepsilon^2 = \frac{1}{n} \cdot \omega(\ln n)$.  This gives
$p(S) \leq (n+1) \cdot  {n^{-\omega(1)}} = {n^{-\omega(1)}}$.
\end{proof}

\begin{corollary}\label{cor:disting}
Any algorithm that makes a polynomial number of oracle queries has 
probability at most $n^{-\omega(1)}$ of distinguishing the functions $f_1$ and $f_2$.
\end{corollary}

We now use these results to establish the hardness of the SSC and SBC 
problems. For concreteness, assume that the output of an approximation algorithm for one of these 
problems consists of a set $S\subseteq V$ as well as the value of the 
objective function on this set.

\begin{theorem} 
\label{thm:ssclb}
The uniform SSC  and the unweighted SBC problems 
(with balance $b=\Theta(1)$) cannot be approximated to 
a ratio  $o\left(\sqrt{\frac{n}{\ln n}}\right)$ 
in the oracle model with polynomial number of queries, even in the case of symmetric functions.
\end{theorem}
\begin{proof}
Suppose for the sake of contradiction that there is a 
polynomial-time $\gamma$-approximation   algorithm for the uniform SSC 
problem, for some $\gamma=o\left(\sqrt{\frac{n}{\ln n}}\right)$, that succeeds 
with high probability.   
We set $\varepsilon=\frac{1}{2\gamma\delta}$ with some $\delta>1$ such that $\beta=\frac{n}{4}(1+\varepsilon)$ is integer. This satisfies 
$\varepsilon^2 = \frac{1}{n}\cdot \omega(\ln n)$.
One feasible solution for the uniform SSC on $f_2$ is the set $R$, with ratio 
$\frac{\beta-n/4}{n^2/4} = \frac{\varepsilon}{n}$. So if the algorithm is given function $f_2$ as input, then with high probability it has to output a set $S$ with ratio $f_2(S)/|S||\bar{S}| \leq \frac{\gamma \varepsilon}{n} = \frac{1}{2\delta n} < \frac{1}{2n}$.
However, for the function $f_1$, the ratio of any set 
is $1/{2\max(|S|,|\bar{S}|)}>\frac{1}{2n}$. So if the algorithm is given $f_1$ as input, its output value differs from the case of $f_2$.
But this contradicts Corollary~\ref{cor:disting}.

For the lower bound to the submodular balanced cut problem, we consider the same two functions $f_1$ and $f_2$ and unit weights. Assuming that there is a $\gamma$-approximation algorithm for SBC, we set $\varepsilon = \frac{2b}{\delta\gamma}$, with $\gamma>1$ ensuring the integrality of $\beta$. This satisfies 
$\varepsilon^2 = \frac{1}{n}\cdot \omega(\ln n)$ if $\gamma=o\left(\sqrt{\frac{n}{\ln n}}\right)$ and $b$ is a constant. Since 
one feasible $b$-balanced cut on $f_2$ is the set $R$, whose function value is 
$\frac{n\varepsilon}{4}$, the algorithm outputs a $b$-balanced set $S$ with 
$f_2(S)\leq {\gamma n \varepsilon}/{4} = bn/2\delta < bn/2$. However, for any $b$, the optimal $b$-balanced cut on $f_1$ is a set of size $bn$, whose function value is ${bn}/{2}$. Thus, given $f_1$, the algorithm would produce a different output, leading to a contradiction.
\end{proof}

\subsection{Algorithm for submodular sparsest cut} \label{subsec:ssc}

Our algorithm for SSC uses a random set $S$ to assign weights to  nodes 
(see Algorithm~\ref{alg:ssc}).  
For each demand pair separated by the set $S$, we add a positive weight equal to its demand $d_i$ to the node that is in $S$, and a negative weight of $-d_i$ to the node that is outside of $S$. This biases the subsequent function minimization to 
separate the demand pairs that are on different sides of  $S$. 

\begin{algorithm}[ht]
  \caption{~~Submodular sparsest cut. ~Input: $V$, $f$, 
$d$, $B$, $p$} \label{alg:ssc} 
\begin{algorithmic}[1] 
\For {$\frac{8n^3}{c} \ln(\frac{1}{1-p})$ iterations}
\State Choose a random set $S$ by including each node $v\in V$ 
independently with probability $\frac{1}{2}$
\State {\bf for} each $v\in V$, initialize a weight $w(v)=0$
\For {each pair  $\{u_i, v_i\}$ with $|\{u_i,v_i\}\cap S|=1$} 
\State Let $s_i \in \{u_i,v_i\}\cap S$ and $t_i \in \{u_i,v_i\}\setminus S$
\Comment{name the unique node in each set}
\State Update weights  $w(s_i) \leftarrow w(s_i) + d_i$; ~$w(t_i) \leftarrow w(t_i) - d_i$
\EndFor
\State Let $\alpha=4\sqrt{\frac{n}{\ln n}}\cdot B$
\State Let $T$ be a subset of $V$ minimizing 
$f(T)-\alpha\cdot \sum_{v\in T} w(v)$ \label{line:sscexp}
\State {\bf if} ~$f(T)-\alpha\cdot \sum_{v\in T} w(v) < 0$,~ {\bf return} 
$T$
\EndFor
\State {\bf return} \emph{fail}
\end{algorithmic}
\end{algorithm}


\begin{lemma} 
\label{lem:negexpr2}
If for some set $T\subseteq V$,  it holds that 
$f(T)-\alpha\cdot \sum_{v\in T} w(v) < 0$, then
$$\frac{f(T)}{\sum_{i:|T\cap \{u_i,v_i\}|=1}d_i} ~<~ \alpha.$$
\end{lemma}
\begin{proof}
We have
\begin{eqnarray*}
\sum_{v\in T} w(v) ~=~ 
\sum_{i:s_i\in T}d_i - \sum_{i:t_i\in T}d_i ~= 
\sum_{i:s_i\in T, t_i\notin T}d_i - \sum_{i:t_i\in T, s_i\notin T}d_i ~\leq
\sum_{i:s_i\in T, t_i\notin T}d_i ~\leq \sum_{i:|T\cap \{u_i,v_i\}|=1}d_i
\end{eqnarray*}
  Now using  the assumption of the lemma we have
\begin{equation}\label{eqn:neg1}
f(T)-\alpha \sum_{i:|T\cap \{u_i,v_i\}|=1}d_i ~~\leq~~ f(T)-\alpha\sum_{v\in T} w(v) ~~<~~ 0.
\end{equation}
Since the function $f$ is non-negative, it must be that 
$\sum_{i:|T\cap \{u_i,v_i\}|=1}d_i > 0$. Rearranging the terms, we get
${f(T)}/{\sum_{i:|T\cap \{u_i,v_i\}|=1}d_i}< \alpha$. 
\end{proof}

\medskip 

Assuming that the input instance is feasible, let $U^*$ be a set
with size $m=|U^*|$, separated demand
$D^* = {\sum_{i:|U^*\cap \{u_i,v_i\}|=1}d_i}$, and value $f(U^*)/D^*<B$.

\begin{lemma}
\label{lem:sscprob}
In one iteration of the outer loop of Algorithm \ref{alg:ssc}, the probability 
that \linebreak 
$\sum_{v\in U^*} w(v) \geq D^* \cdot \frac{1}{4}\sqrt{\frac{\ln n}{n}}$ 
~is at least $\frac{c}{8 n^3}$.
\end{lemma}
\begin{proof}
Let $\varepsilon = \sqrt{\frac{\ln n}{n}}$. We denote by
${\mathcal A}$ the event that 
$|U^* \cap S| \geq \frac{m}{2}\left(1+\varepsilon\right)$, where $S$ is the 
random set chosen by Algorithm \ref{alg:ssc}, and bound the above
probability by the following product: 

$$\Pr\left[\sum_{v\in U^*} w(v) \geq \frac{\varepsilon}{4} D^* \right] ~\geq~ 
\Pr\left[\sum_{v\in U^*} w(v) \geq \left.\frac{\varepsilon}{4} D^* \, \right| {\mathcal A} \, \right] 
\cdot \Pr[{\mathcal A}].$$
We observe that by Theorem \ref{thm:gen}, the  probability of ${\mathcal A}$ is at least $\frac{c}{2} n^{-5/2}$. All the probabilities and expectations in the rest of the proof are conditioned on the event ${\mathcal A}$.  

Let us now consider the expected value of $\sum_{v\in U^*} w(v)$.
Fix a particular demand pair $\{u_i, v_i\}$ that is separated by the optimal solution, and assume without loss of generality that $u_i\in U^*$ and $v_i\notin U^*$. Let $p_u$ be the probability that $u_i\in S$, and $p_v$ be the probability that 
$v_i\in S$.  Then $p_u = \frac{|U^* \cap S|}{|U^*|}\geq (1+\varepsilon)/2$,
$p_v = \frac{1}{2}$, and the two events are independent. So
\begin{eqnarray*}
\Pr[u_i = s_i] &=& \Pr[u_i\in S \wedge v_i\notin S]
~=~ p_u\cdot (1-p_v) ~\geq~ (1+\varepsilon)/4, \\
\Pr[u_i = t_i] &=& \Pr[u_i\notin S \wedge v_i\in S] ~=~ (1-p_u)\cdot p_v ~\leq~ 
(1-\varepsilon)/4.
\end{eqnarray*}
Then the expected contribution of this demand pair to $\sum_{v\in U^*} w(v)$ is 
equal to 
$$\Pr[u_i=s_i]\cdot d_i + \Pr[u_i=t_i]\cdot (-d_i) ~\geq~ d_i\cdot \frac{\varepsilon}{2}.$$
By linearity of expectation,
$${\rm E}\left[\sum_{v\in U^*} w(v)\right] ~\geq ~ D^* \cdot \frac{\varepsilon}{2}.$$
We now use Markov's inequality \cite{motwani:raghavan} to bound the desired 
probability.  For this we define a nonnegative random  variable 
$Y=D^*-\sum_{v\in U^*} w(v)$. Then  
${\rm E}[Y]\leq (1-\varepsilon/2) D^*$. So 
$$\Pr\left[\sum_{v\in U^*} w(v) \leq \frac{\varepsilon}{4}D^* \right] ~= ~
\Pr\left[Y\geq (1-\frac{\varepsilon}{4})D^*\right] ~\leq~
 \frac{{\rm E}[Y]}{(1-\varepsilon/4) D^*} ~\leq~
 \frac{1-\varepsilon/2}{1-\varepsilon/4} ~=~
1-\frac{\varepsilon}{4-\varepsilon} ~\leq~ 1-\frac{\varepsilon}{4}
$$
It follows that 
$$\Pr\left[\sum_{v\in U^*} w(v) \geq \frac{\varepsilon}{4}D^* \right] ~\geq~ 
\frac{\varepsilon}{4} ~=~ \frac{1}{4}\sqrt{\frac{\ln n}{n}} ~\geq~ 
\frac{1}{4\sqrt{n}},$$ 
concluding the proof of the lemma.
\end{proof}

\begin{theorem}
For any feasible instance of SSC problem, Algorithm \ref{alg:ssc} returns a 
solution of cost at most 
$4 \sqrt{\frac{n}{\ln n}}\cdot B$, with probability at least $p$. 
\end{theorem}
\begin{proof}
By Lemma \ref{lem:sscprob}, the inequality 
$\sum_{v\in U^*} w(v) \geq D^* \cdot \frac{1}{4}\sqrt{\frac{\ln n}{n}}$
holds with probability at least $c/8n^3$ in each iteration. Then the 
probability that it holds in  any of the $\frac{8n^3}{c} \ln(\frac{1}{1-p})$ 
iterations is at least $p$. Now,  assuming that it does   hold, 
the algorithm finds a set $T$ such that 
$$ f(T)-\alpha\cdot \sum_{v\in T} w(v) ~\leq ~  
f(U^*)-\alpha\cdot \sum_{v\in U^*} w(v) ~\leq~ 
 f(U^*)- \left( 4\sqrt{\frac{n}{\ln n}} \cdot B \right) \left(D^* \cdot 
\frac{1}{4}\sqrt{\frac{\ln n}{n}}\right) ~<~ 0.$$
Applying Lemma \ref{lem:negexpr2}, we get that
${f(T)}/{\sum_{i:|T\cap \{u_i,v_i\}|=1}d_i} < \alpha =
  4\sqrt{\frac{n}{\ln n}} \cdot B$, 
which means that $T$ is the required approximate solution.
\end{proof}

\subsection{Submodular balanced cut}
For submodular balanced cut, we use as a subroutine
the weighted SSC problem that  can be approximated to a factor 
$\gamma = O\left(\sqrt{\frac{n}{\ln n}}\right)$ using Algorithm~\ref{alg:ssc}.
This allows us  to obtain a bicriteria 
approximation for SBC in a similar way that Leighton and Rao 
\cite{leighton:rao} use their algorithm for sparsest cut on 
graphs to approximate balanced cut on graphs. 
Leighton and Rao 
present two versions of an algorithm for the balanced 
cut problem on graphs --- one for undirected graphs, and one for directed 
graphs. The algorithm for undirected graphs has a better balance guarantee.
We describe adaptations of these algorithms to the submodular version of the 
balanced cut problem. Our first algorithm extends the one for undirected 
graphs, and it works for symmetric submodular functions. For a given $b'\leq 1/3$, it 
finds a $b'$-balanced cut whose cost is within a factor 
$O\left(\frac{\gamma}{b-b'}\right)$ of the cost of any $b$-balanced cut, for 
$b'<b\leq \frac{1}{2}$.  
The second algorithm works for arbitrary non-negative submodular functions and 
produces a ${b'}/{2}$-balanced cut of cost within 
$O\left(\frac{\gamma}{b-b'}\right)$ of any $b$-balanced cut, for any $b'$ and 
$b$ with $b'<b\leq 1/2$.

\subsubsection{Algorithm for symmetric functions}
The algorithm for SBC on symmetric functions (Algorithm \ref{alg:ssbc}) 
repeatedly finds approximate weighted submodular sparsest cuts 
$(S_i,\bar{S}_i)$ and collects  
their smaller sides into the set $T$, until $(T,\bar{T})$ becomes $b'$-balanced. 
The algorithm and analysis basically follow Leighton and Rao 
\cite{leighton:rao}, with the main difference being that instead of removing 
parts of the graph, we set the weights of the corresponding elements to zero. 
Then the obtained sets $S_i$ are not necessarily disjoint.

\begin{algorithm}[ht] 
  \caption{~~Submodular balanced cut for symmetric functions. ~Input: $V$, $f$, 
$w$, $b'\leq \frac{1}{3}$} \label{alg:ssbc} 
\begin{algorithmic}[1]
\State Initialize $w'=w$, $i=0$, $T=\emptyset$
\While{$w'(V) > (1-b')w(V)$} 
\State Let $S$ be a $\gamma$-approximate weighted SSC on $V$, $f$, and 
weights $w'$ \label{line:ssbc0}
\State Let $S_i=\argmin(w'(S),w'(\bar{S}))$; ~$w'(S_i)\leftarrow 0$; 
~$T\leftarrow T\cup S_i$; ~$i\leftarrow i+1$ \label{line:ssbc1}
\EndWhile
\State {\bf return} $T$
\end{algorithmic}
\end{algorithm}

\begin{theorem}\label{thm:ssbc}
If the system $(V,f,w)$, where $f$ is a symmetric submodular function, contains a $b$-balanced cut of cost $B$, then  
Algorithm \ref{alg:ssbc} finds a $b'$-balanced cut $T$ with  
$f(T) = O\left(\frac{B}{b-b'} \sqrt{\frac{n}{\ln n}}\right)$, for a given 
$b'<b$, $b'\leq \frac{1}{3}$.
\end{theorem}
\begin{proof}
The algorithm terminates in $O(n)$ iterations, since the weight of at least 
one new element is set to zero on line \ref{line:ssbc1} (otherwise the 
solution to SSC found on line \ref{line:ssbc0} would have infinite cost).

Now we consider $w(T)$. By the termination condition of the while loop, we 
know that when it exits, $w'(V)\leq (1-b')w(V)$, which means that $w'$ has 
been set to zero for elements of total weight at least $b'w(V)$. But those are 
exactly the elements in $T$, so $w(T)\geq b'w(V)$.  Now consider the last 
iteration of the loop. At the beginning of this iteration, we have 
$w'(V)>(1-b')w(V)$, which means that at the end of it we have 
$w'(V)>\frac{1}{2}(1-b')w(V)$, because the weight of the smaller (according to 
$w'$) of $S$ or  
$\bar{S}$ is set to zero. But $w'(V)$ at the end of the algorithm is exactly 
the weight of $\bar{T}$, which means that 
$w(\bar{T})>\frac{1}{2}(1-b')w(V)\geq \frac{1}{3}w(V)\geq b'w(V)$, using the 
assumption $b'\leq 1/3$ twice. So the cut $(T,\bar{T})$ is $b'$-balanced.

Suppose that $U^*$ is a $b$-balanced cut with $f(U^*)=B$. 
In any iteration $i$ of the while loop, we know 
that two inequalities hold: $w'(U^*)+w'(\bar{U}^*)>(1-b')w(V)$ (by the 
loop condition), and $\max(w'(U^*), w'(\bar{U}^*))\leq (1-b)w(V)$ (by $b$-balance). 
Given these inequalities, the minimum value that the product 
$w'(U^*)\cdot w'(\bar{U}^*)$ can have is $(b-b')w(V)\cdot (1-b)w(V)$.
So with weights $w'$, there is a solution to the SSC problem with value 
$$\frac{f(U^*)}{w'(U^*) w'(\bar{U}^*)} ~\leq~ \frac{B}{(b-b')w(V)\cdot (1-b)w(V)},$$
and the set $S_i$ found by the $\gamma$-approximation algorithm satisfies
$$\frac{f(S_i)}{w'(S_i) w'(\bar{S}_i)} ~\leq~ \frac{\gamma B}{(b-b')w(V)\cdot (1-b)w(V)}.$$
Since in iteration $i$, $w'(S_i)=w(S_i\setminus \bigcup_{j=0}^{i-1}S_j)$, 
$w'(\bar{S}_i)\leq w(V)$, and $(1-b)\geq 1/2$,
$$f(S_i)~\leq~ w(S_i\setminus \bigcup_{j=0}^{i-1}S_j) \frac{2B\gamma}{(b-b')w(V)}.$$
Now $f(T)\leq \sum_{i}f(S_i) \leq w(T)\cdot 2B\gamma/(b-b')w(V) =
 B \cdot O(\frac{\gamma}{b-b'})$. 
\end{proof}

\subsubsection{Algorithm for general functions}
The algorithm for general functions (Algorithm \ref{alg:sbc}) also repeatedly 
finds weighted submodular sparsest cuts $(S_i,\bar{S}_i)$, but it uses them to collect two 
sets: either it puts $S_i$ into $T_1$, or it puts $\bar{S}_i$  into  $T_2$.
Thus, the values of $f(T_1)$ and  $\bar{f}(T_2)$ can be bounded using the 
guarantee of the SSC algorithm (where $\bar{f}(S)=f(\bar{S})$).

\begin{algorithm}[ht] 
  \caption{~~Submodular balanced cut. ~Input: $V$, $f$, $w$, $b'$} \label{alg:sbc} 
\begin{algorithmic}[1]
\State Initialize $w'=w$, $i=0$, $T_1=T_2=\emptyset$
\While{$w'(V) > (1-b')w(V)$} 
\State Let $S_i$ be a $\gamma$-approximate weighted SSC on $V$, $f$, and 
weights $w'$
\State {\bf if} $w'(S_i) \leq w'(\bar{S}_i)$ 
{\bf then} set 
$T_1\leftarrow T_1\cup S_i$; ~$w'(S_i)\leftarrow 0$; ~$i\leftarrow i+1$
\State {\bf else} set $T_2\leftarrow T_2\cup \bar{S}_i$; 
~$w'(\bar{S}_i)\leftarrow 0$; ~$i\leftarrow i+1$ 
\EndWhile
\State {\bf if} $w(T_1)\geq w(T_2)$  {\bf then return} $T_1$ {\bf else return} $\bar{T}_2$
\end{algorithmic}
\end{algorithm}

\begin{theorem}
If the system $(V,f,w)$ contains a $b$-balanced cut of cost $B$, then 
Algorithm \ref{alg:sbc} finds a $b'/2$-balanced cut $T$ with  
$f(T) = O\left(\frac{B}{b-b'} \sqrt{\frac{n}{\ln n}}\right)$, for a given 
$b'< b$.
\end{theorem}
\begin{proof}
When the while loop exits, $w'(V) \leq (1-b')w(V)$, so the total weight of 
elements in $T_1$ and $T_2$ (the ones for which $w'$ has been set to zero) is 
at least $b'w(V)$. So $\max(w(T_1), w(T_2))\geq b'w(V)/2$. 
At the beginning of the last iteration of the loop, $w'(V)>(1-b')w(V)$. Since 
the weight of the smaller of $S_i$ and $\bar{S}_i$ is set to zero, at the end 
of this iteration $w'(V)>\frac{1}{2}(1-b')w(V)$. 
Let $T$ be the set output by the algorithm. Since $w'(T)=0$, we have 
$w(\bar{T})\geq w'(V)>\frac{1}{2}(1-b')w(V)\geq b'/2$, using $b'\leq 1/2$. 
Thus we have shown that Algorithm \ref{alg:sbc} outputs a $b'/2$-balanced cut. 

The function values can be bounded as
$f(T_1)=B \cdot O(\frac{\gamma}{b-b'})$ and 
$\bar{f}(T_2)=B \cdot O(\frac{\gamma}{b-b'})$ using a proof similar to that of 
Theorem \ref{thm:ssbc}. 
\end{proof}


\section{Submodular minimization with cardinality lower bound} 
\label{sec:sk}

We start with the lower bound result.
Let $R$ be a random subset of $V$ of size $\alpha = \frac{x\sqrt{n}}{5}$, let 
$\beta=\frac{x^2}{5}$, and $x$ be any parameter satisfying 
$x^2= \omega(\ln n)$ and such that $\alpha$ and $\beta$ are integer.  We use the following two monotone submodular functions:
\begin{eqnarray}
\label{eq:learnlb}
f_3(S) ~=~ \min\left(|S|,\,\alpha\right), & \hspace{1cm} &
f_4(S) ~=~ \min\left( \beta+|S\cap \bar{R}|,~|S|,~\alpha \right).
\end{eqnarray}

\begin{lemma} \label{lem:disting3}
Any algorithm that makes a polynomial number of oracle queries has 
probability $n^{-\omega(1)}$  of distinguishing the functions $f_3$ and $f_4$ above.
\end{lemma}
\begin{proof}
By Lemma \ref{lem:alg}, it suffices to prove that for any set $S$, the 
probability that $f_3(S)\neq f_4(S)$ is at most $n^{-\omega(1)}$. It is easy to check (similarly to the proof of Lemma \ref{lem:disting}) that 
$\Pr[f_3(S)\neq f_4(S)]$ is maximized for sets $S$ of 
size $\alpha$. And for a set $S$ with $|S|=\alpha$, $f_3(S)\neq f_4(S)$ if and 
only if $\beta+|S\cap \bar{R}|< |S|$, or, equivalently, $|S\cap R|>\beta$. So 
we analyze the probability that $|S\cap R|>\beta$.

$R$ is a random subset of $V$ of size $\alpha$.  Let us consider a different 
set, $R'$, which is obtained by independently including each element of $V$ 
with probability $\alpha/n$. The expected size of $R'$ is $\alpha$, and the probability 
that $|R'|=\alpha$ is at least $1/(n+1)$. Then
$$\Pr\left[|S\cap R| > \beta\right] ~=~ 
\Pr\left[|S\cap R'| > \beta ~\left|~ |R'|=\alpha \right.\right] ~\leq~
(n+1)\cdot \Pr\left[|S\cap R'| > \beta \right],$$
and it suffices to show that $\Pr\left[|S\cap R'| > \beta \right] = n^{-\omega(1)}$. 
For this, we use Chernoff bounds. The expectation of $|S\cap R'|$ is 
$\mu = \alpha|S|/n = \alpha^2/n = x^2/25$. Then $\beta = 5\mu$. Let $\delta=4$. Then 
$$\Pr\left[|S\cap R'| > (1+\delta) \mu \right] ~<~ 
\left(\frac{e^\delta}{(1+\delta)^{1+\delta}}\right)^\mu ~=~ 
\left(\frac{e^4}{5^5}\right)^{\frac{x^2}{25}} ~\leq~ 0.851^{x^2}.$$
Since $x^2= \omega(\ln n)$, we get that this probability is $n^{-\omega(1)}$.
\end{proof}

\begin{theorem}
\label{thm:sklb}
There is no $(\rho,\sigma)$ bicriteria approximation algorithm
for the SML problem, even with monotone functions, for any $\rho$ and $\sigma$ with 
$\frac{\rho}{\sigma} = o\left(\sqrt{\frac{n}{\ln n}}\right)$.
\end{theorem}
\begin{proof}
We assume that any algorithm for this problem outputs a set of elements as 
well as the function value on this set.  Suppose that a bicriteria algorithm with 
$\frac{\rho}{\sigma} = o\left(\sqrt{\frac{n}{\ln n}}\right)$ exists.
Let $f_3$ and $f_4$ be the two monotone functions in (\ref{eq:learnlb}), with 
$x=\frac{\sigma\sqrt{n}}{\delta\rho}$, where $\delta>1$ is a constant that ensures that $\alpha$ and $\beta$ are integer. Then $x$ satisfies $x^2=\omega(\ln n)$.  
Consider the output of the algorithm when given $f_4$ as input and
$W=\alpha$. The optimal solution in this case is the set $R$, with $f(R)=\beta$. 
So the algorithm finds an approximate solution 
 $T$ with $f_4(T)\leq \rho \beta$ and $|T|\geq \sigma \alpha$. 
However, we show that no set $S$ with $f_3(S)\leq \rho \beta$ and $|S|\geq \sigma \alpha$ exists, which means that if the input is the function $f_3$, then the algorithm produces a different answer, thus distinguishing $f_3$ and $f_4$. We assume for contradiction that such a set $S$ exists and  consider two cases. First, suppose $|S|\geq \alpha$.
Then $f_3(S)\leq \rho\beta = \frac{\sigma\sqrt{n}}{\delta x} \frac{x^2}{5} = 
\frac{\sigma\alpha}{\delta} < \alpha$, since $\delta>1$ and by definition 
$\sigma\leq 1$. But this is a contradiction because $f_3(S)=\alpha$ for all $S$ 
with $|S|\geq \alpha$. The second case is $|S|<\alpha$. Then we have 
$|S|\geq \sigma\alpha$ and 
$f_3(S)\leq \rho\beta = \frac{\sigma \alpha}{\delta} \leq |S|/\delta$, 
which is also a contradiction because $|S|\geq \sigma\alpha >0$ and $f_3(S)=|S|$ for $|S|<\alpha$.
\end{proof}

\subsection{Algorithm for SML}
Our relaxed decision procedure for the SML problem with weights 
$\{0,1\}$ (Algorithm \ref{alg:sk}) builds up the solution out of multiple 
sets that it finds using submodular function minimization. 
If the weight requirement $W$ is larger than half the total weight 
$w(V)$, then collecting sets whose ratio of function value to weight of new 
elements is low (less than $2B/W$), until a total weight of at least $W/2$ is 
collected, finds the required approximate solution. In the other case, if $W$ 
is less than $w(V)/2$,  
the algorithm looks for sets $T_i$ with low ratio of  function value to the 
weight of new elements in the intersection of $T_i$ and a random set 
$S_i$. These sets not only have small $f(T_i)/w(T_i)$ 
ratio, but also have bounded function value $f(T_i)$. If such a set 
is found, then it is added to the solution. 

\begin{algorithm}[ht] 
  \caption{~~SML. ~Input: $V$, $f$, $w:V\rightarrow \{0,1\}$, $W$, 
$B$, $p$} \label{alg:sk} 
\begin{algorithmic}[1]
\State Initialize $U_0=\emptyset$; $i=0$
\If{$W\geq w(V)/2$}
\Comment{ case $W \geq \frac{w(V)}{2}$}
\While{$w(U_i) < W/2$} \label{line:wh1}
\State Let $T_i$ be a subset of $V$ minimizing $f(T)-\frac{2B}{W} \cdot w(T \setminus U_i)$ 
\label{line:expr1}
\State {\bf if} {$f(T_i) <\frac{2B}{W}\cdot w(T_i\setminus U_i)$} 
{\bf then} Let  $U_{i+1}=U_i \cup T_i$;~ $i=i+1$ 
{\bf else return} \emph{fail}
\EndWhile
\State {\bf return} $U = U_i$ 
\EndIf
\State Let $\alpha= \frac{2B}{W} \sqrt{\frac{n}{\ln n}}$ 
\Comment{ case $W<\frac{w(V)}{2}$}
\While{$w(U_i)<W/2$} \label{line:wh2}
\State Choose a random $S_i\subseteq V\setminus U_i$, including each element 
with probability  $\frac{W}{w(V)}$ 
\State Let $T_i$ be a subset of $V$ minimizing 
$f(T)-\alpha\cdot w(T\cap S_i)$ \label{line:expr}
\State {\bf if} $f(T_i) \leq \alpha\cdot w(T_i\cap S_i)$ and 
$f(T_i)\leq 4B \sqrt{\frac{n}{\ln n}}$ 
~{\bf then } Let $U_{i+1}=U_i \cup T_i$;~ $i=i+1$  \label{line:cond}
\State {\bf if} the number of iterations exceeds 
$\frac{3n^{9/2}}{c}\ln \left(\frac{n}{1-p}\right)$,~ 
{\bf return} \emph{fail} 
\EndWhile
\State \Return $U = U_i$
\end{algorithmic}
\end{algorithm}

\begin{theorem}
\label{thm:skalg}
Algorithm \ref{alg:sk} is a $(5\sqrt{\frac{n}{\ln n}}, \frac{1}{2})$ 
bicriteria decision procedure for the SML problem. That is, given a feasible 
instance, it outputs a set $U$ with $f(U)\leq 5\sqrt{\frac{n}{\ln n}} B$ and 
$w(U)\geq W/2$ with probability at least $p$.
\end{theorem}
\begin{proof}
Assume that the instance is feasible, and let $U^*\subseteq V$ be a set with 
$w(U^*)\geq W$ and $f(U^*) < B$. We consider two cases, $W\geq w(V)/2$ and 
$W <w(V)/2$, which the algorithm handles separately. 

First, assume that $W\geq w(V)/2$ and consider one of 
the iterations of the while loop on line~\ref{line:wh1}. By the loop 
condition, $w(U_i)<W/2$, so  $w(U^*\setminus U_i) > W/2$. As a result, for the 
set $U^*$, the expression  on line \ref{line:expr1}  is negative:
$$f(U^*) - \frac{2B}{W}\cdot w( U^* \setminus U_i) ~< ~ f(U^*) - B ~< ~0.$$
Then for the set $T_i$ which minimizes this expression, it would 
also be negative, implying that $w(T_i\setminus U_i)$ is positive, and so $w(U_i)$ 
increases in each iteration. As a result, if the instance is feasible, then 
after at most $n$ iterations of the 
loop on line \ref{line:wh1}, a set $U$ is found with $w(U)\geq W/2$. For 
the function value, we have 
$$f(U) ~\leq~ \sum_i f(T_i) ~<~ \frac{2B}{W} \sum_i w(T_i\setminus U_i) ~\leq~ 
\frac{2B}{W} \cdot w(V) ~\leq ~ 4B$$
by our assumption about $W$.

The second case is $W<w(V)/2$. Assuming Claim \ref{cl:int} below, which is proved 
later, we show that  in each 
iteration of the while loop on line \ref{line:wh2}, with probability at least 
 $\frac{c}{3n^{7/2}}$,  a new non-empty set $T_i$ is added to $U$. 
This implies that after $\frac{3n^{9/2}}{c}\ln \left(\frac{n}{1-p}\right)$ 
iterations, the loop successfully terminates with probability at least $p$.

\begin{claim} \label{cl:int}
In each iteration of the while loop on line \ref{line:wh2} of Algorithm 
\ref{alg:sk}, both of the following two inequalities hold with probability at least 
$\frac{c}{3n^{7/2}}$.
\begin{equation} \label{eq:si}
w(U^* \cap S_i) ~>~ \frac{B}{\alpha}~=~ \frac{W}{2} \sqrt{\frac{\ln n}{n}} 
 ~~~~~{\rm and}~~~~~ w(\bar{U}^* \cap S_i) ~\leq~ 1.5 W.
\end{equation}
\end{claim}

We show that if inequalities (\ref{eq:si}) hold, then the set $T_i$ found by 
the algorithm on line \ref{line:expr} is non-empty and satisfies the 
conditions on line~\ref{line:cond}, which means that new elements are added to 
$U$. Since $T_i$ is a minimizer of the expression on line 
\ref{line:expr}, and using (\ref{eq:si}),
$$f(T_i)- \alpha \cdot w(T_i\cap S_i)~\leq ~ 
f(U^*) - \alpha \cdot w(U^* \cap S_i) ~<~
f(U^*) -B ~<~ 0,$$
which means that $T_i$ satisfies the first condition on line \ref{line:cond} 
and  is non-empty. Moreover, from the same inequality and the second part of 
(\ref{eq:si}) we have  
$$f(T_i) ~\leq~ f(U^*) + \alpha\cdot (w(T_i\cap S_i) - w(U^* \cap S_i)) ~\leq~
B + \alpha \cdot w(\bar{U}^* \cap S_i) ~\leq ~ B+1.5\alpha W ~\leq~ 
4B\sqrt{\frac{n}{\ln n}},$$
which means that $T_i$ also satisfies the second condition on line 
\ref{line:cond}. 

Now we analyze the function value of the set output by the algorithm.
Let $T_i$ be the last set added to $U$ by the while loop, and consider the set 
$U_i$ just before $T_i$ is added to it to produce $U_{i+1}$. By the loop 
condition, we have $w(U_i)<W/2$. Then, by submodularity and condition on 
line \ref{line:cond},
$$f(U_i)~\leq~ \sum_{j=0}^{i-1} f(T_j) ~\leq~
 \sum_{j=0}^{i-1} \alpha \cdot w(T_j \cap S_j) ~\leq~
\alpha \cdot w(U_i) ~<~ \alpha\cdot\frac{W}{2} ~=~ B\sqrt{\frac{n}{\ln n}}.$$
So for the set $U$ that the algorithm outputs, 
$f(U)\leq f(U_i)+f(T_i) \leq 5B \sqrt{\frac{n}{\ln n}}$. And by the exiting 
condition of the while loop, $w(U)\geq W/2$.
\end{proof}

\medskip

\begin{proof}[Proof of Claim \ref{cl:int}.]
Because the events corresponding to the two inequalities are independent, we 
bound their  
probabilities separately and then multiply. To bound the probability of 
the first one let  $m=w(U^*\setminus U_i)$ be the number of elements of $U^*$ 
with  weight 1 that are  in $V\setminus U_i$. Since $w(U^*)\geq W$  
and $w(U_i)<W/2$ by the condition of the loop, we have $m>W/2$. 
We invoke Theorem \ref{thm:gen} with parameters $m$, $q=W/w(V)$, and 
$\varepsilon=\frac{w(V)}{2m}\sqrt{\frac{\ln n}{n}}$.
To ensure that $\varepsilon< \frac{1-q}{q}$ and this theorem can be 
applied, we assume that $n\geq 9$, so that $\sqrt{\ln n/n} < 1/2$, and get
$$\frac{1-q}{q}-\varepsilon ~=~ 
\frac{w(V)}{W} -1 - \frac{w(V)}{2m} \sqrt{\frac{\ln n}{n}} ~>~
\frac{w(V)}{2W} - 1 ~>~ 0.$$
Thus the inequality 
$w(U^*\cap S_i) \geq \lceil qm(1+\varepsilon)\rceil > qm\varepsilon =
\frac{W}{2}\sqrt{\frac{\ln n}{n}}$ holds with  probability at least 
(simplifying using inequalities $w(V)-W\geq w(V)/2$, $w(V)\leq n$, and 
$1\leq W<2m$)
$$c\,q\, m^{-\frac{3}{2}} \exp\left[\frac{-\varepsilon^2 qm}{1-q}\right] ~=~
c\,\frac{W}{w(V)} \, m^{-\frac{3}{2}} \exp\left[-\frac{w(V)^3 ~W~m~ \ln n}
{4m^2~n~ w(V)~(w(V)-W)}\right]  ~\geq~ cn^{-7/2}.$$

For the second inequality,  we notice that the expectation of 
$w(\bar{U}^* \cap S_i)$ is $w(\bar{U}^*) \cdot \frac{W}{w(V)} \leq W$. 
So by Markov's inequality, the probability that 
$w(\bar{U}^* \cap S_i) \leq 1.5 W$ is at least $1/3$.   
\end{proof}


\section{Submodular load balancing}

\subsection{Lower bound}
We give two monotone submodular functions that are hard to distinguish, 
but whose value of the optimal solution to the SLB problem 
differs by a large factor.  These functions are:
\begin{eqnarray} \label{fns:slb}
f_5(S)~=~\min\left(|S|,\,\alpha\right) & \hspace{1cm} &
f_6(S)~=~\min\left(\sum_i \min\left(\beta, |S\cap V_i|\right),~\alpha\right).
\end{eqnarray}
Here $\{V_i\}$ is a random (unknown to the algorithm) 
partition of $V$ into $m$ equal-sized sets.  We set 
$m=\frac{5\sqrt{n}}{x}$, $\alpha=\frac{n}{m}=\frac{x\sqrt{n}}{5}$, 
$\beta=\frac{x^2}{5}$, with any parameter $x$ satisfying $x^2= \omega(\ln n)$, and values chosen so that $\alpha$ and $\beta$ are integer.
 
\begin{lemma} \label{lem:disting2}
Any algorithm that makes a polynomial number of oracle queries has 
probability $n^{-\omega(1)}$  of distinguishing the functions $f_5$ and $f_6$ above.
\end{lemma}
\begin{proof}
By Lemma \ref{lem:alg}, it suffices to bound the probability, over
the random choice of the sets $V_i$, that $f_5(S)\neq f_6(S)$ for any one set  
$S$. Since $f_5 \geq f_6$, this is the same as $\Pr\left[f_6(S)-f_5(S)<0\right]$.
First, we show that this probability is maximized when $|S|=\alpha$.  For 
$|S|\geq \alpha$, 
\begin{eqnarray*}
\Pr\left[f_6(S)-f_5(S)<0\right] 
&=& \Pr\left[\sum_i \min\left(\beta, |S\cap 
V_i|\right) < \alpha \right],
\end{eqnarray*}
and since the sum in this expression can only decrease if an element is removed from $S$, we have that for $|S| \geq \alpha$, this 
probability is maximized at $|S|=\alpha$.
For $|S|\leq \alpha$, 
\begin{eqnarray*}
\Pr\left[f_6(S)-f_5(S)<0\right] 
&=& \Pr\left[\min\left(\sum_i \min\left(\beta, |S\cap 
V_i|\right),\alpha\right) - |S| <0\right]\\
&=& \Pr\left[\min\left(\sum_i \min\left(\beta, |S\cap V_i|\right) - \sum_i 
|S\cap V_i|,~\alpha- |S| \right)  <0\right]\\ 
&=& \Pr\left[\sum_i \min\left(\beta - |S\cap V_i|, 0 \right) <0\right].
\end{eqnarray*}
Since the sum in this expression can only decrease if an element is 
added to $S$, we have that for $|S| \leq \alpha$, the 
probability is maximized at $|S|=\alpha$.

So suppose that $|S|=\alpha$.  We notice that if for all $i$, 
$|S\cap V_i|\leq \beta$, then $f_5(S)=f_6(S)$.  Therefore, a necessary 
condition for the two functions to be different is that $|S\cap V_i|>\beta$  
for some $i$.  
Since $V_1$ is a random  subset of $V$ of size $\alpha$, we can use the same 
calculation as in the proof of Lemma \ref{lem:disting3} to show that 
$\Pr\left[|S\cap V_1|>\beta\right] \leq n^{-\omega(1)}$.
Applying the union bound, we get 
that the probability that  $|S\cap V_i|>\beta$ for \emph{any} $i$ is also 
$n^{-\omega(1)}$. 
\end{proof}

\begin{theorem} \label{thm:slbhard}
The SLB problem is hard to approximate to a factor of 
$o\left(\sqrt{\frac{n}{\ln n}}\right)$. 
\end{theorem}
\begin{proof}
Suppose that there is a $\gamma$-approximation algorithm for SLB, where 
$\gamma = o\left(\sqrt{\frac{n}{\ln n}}\right)$. Let $x=\sqrt{n}/\delta\gamma$, where $\delta>1$ is such that $\alpha$ and $\beta$ are integer. This satisfies $x^2= \omega(\ln n)$. Now consider running the algorithm with the 
input function $f_6$ and size of partition $m$.  For this input, partition  
$\{V_i\}$ constitutes the optimal solution whose value is 
$f_6(V_i) = \beta$, so the algorithm returns a solution whose value is at 
most $\gamma\beta = \alpha/\delta$. 
However, for the input $f_5$ and $m$, any partition must 
contain a set $S$ with size  $|S|\geq n/m =\alpha$ (since this is the average 
size). For this set, the function value is $f_5(S)=\alpha>\alpha/\delta$. This 
means that for $f_5$ the algorithm produces a different answer than for $f_6$, 
which contradicts Lemma \ref{lem:disting2}. 
\end{proof}

\subsection{Algorithms for SLB}
We note that the technique of Svitkina and Tardos 
\cite{svitkina:tardos:multiway-cuts}  used for min-max multiway cut can be 
applied to the \emph{non-uniform} SLB problem to   
obtain an $O(\sqrt{n \log n})$ approximation algorithm, using the 
approximation algorithm for the SML problem presented in Section \ref{sec:sk} 
as a subroutine. Also, an $O(\sqrt{n} \log n)$ approximation for the non-uniform SLB appears in \cite{goemans:learning}.

In this section we present two algorithms, with improved 
approximation ratios, for the \emph{uniform} SLB problem. We begin by 
presenting a very simple  algorithm that gives a 
$\min(m, \left\lceil{\frac{n}{m}}\right\rceil) = O(\sqrt{n})$ approximation.  
Then we give a more complex algorithm that improves the approximation ratio to 
$O\left(\sqrt{\frac{n}{\ln n}}\right)$, thus matching the lower bound. Our 
first algorithm simply partitions the elements into $m$ sets of roughly equal 
size. 

\begin{theorem}
\label{thm:simplelb}
The algorithm that partitions the elements into $m$ arbitrary sets of size at most 
$\left\lceil{\frac{n}{m}}\right\rceil$ each is a 
$\min(m, \left\lceil{\frac{n}{m}}\right\rceil)$ approximation for the 
SLB problem.  
\end{theorem}
\begin{proof}
Let $\{U_1^*,...,U_m^*\}$ denote the optimal solution with value $B$, and let 
$A$ be the value of  the solution $\{S_1,...,S_m\}$ found by the algorithm.  
We exhibit two lower bounds on $ B$ and two upper bounds on $A$, and then 
establish the approximation ratio by comparing these bounds. 
For the lower bounds on $ B$, we claim that $B \geq \max_{j\in V} f(\{j\})$ 
and $B\geq f(V)/m$. For the first one, let $j$ be the element maximizing 
$f(\{j\})$, and let $U_i^*$ be the set in the optimal solution that contains 
$j$. Then $B \geq f(U_i^*)\geq f(\{j\})$ by monotonicity. For the second 
bound, by submodularity we have that $f(V)\leq \sum_i f(U_i^*) \leq m B$.
To bound $A$, we notice that $A\leq f(V)$ (by monotonicity), and that 
$A \leq  \left\lceil{\frac{n}{m}}\right\rceil \max_{j\in V} f(\{j\})$, since 
each set $S_i$ contains at most  
$\left\lceil{\frac{n}{m}}\right\rceil$ elements, and 
$f(S_i)\leq \sum_{j\in S_i} f(\{j\})$. Comparing with the lower bounds on $B$, 
we get the result.
\end{proof}

\medskip 

\begin{algorithm}[ht] 
  \caption{~~Submodular load balancing. 
~Input: $V$, $m>\sqrt{\frac{n}{\ln n}}$, monotone $f$, $B$, $p$}\label{alg:slb} 
\begin{algorithmic}[1]
\State {\bf if} for any $v\in V$, ~$f(\{v\})\geq B$,~ 
{\bf  return} \emph{fail}  \label{line:check}
\State Let $\alpha={Bm}/{\sqrt{n\ln n}}$; ~Initialize $V'=V$, ~$i=0$
\While{$|V'| > m\sqrt{\frac{n}{\ln n}}$}
\State Choose a random $S\subseteq V'$, including each element independently 
with probability  $\frac{n}{m|V'|}$ \label{line:prob}
\If{$|S|\leq 2\frac{n}{m}$}
\State Let $T\subseteq S$ be a subset minimizing $f(T)-\alpha\cdot |T|$ \label{step:slbmin}
\State {\bf if}~{$f(T)-\alpha\cdot |T| < 0$ } {\bf then} 
set $T_i = T$;~~$i=i+1$;~~$V' = V'\setminus T$
\EndIf
\State {\bf if} the number of iterations exceeds 
$\frac{2n^3}{c} \ln(\frac{n}{1-p})$,~ \Return \emph{fail}
\EndWhile
\State Let ${\mathcal T}$ be the collection of sets $T_i$ produced by the 
while loop
\State Partition ${\mathcal T}$ into $m$ groups 
${\mathcal T}_1,..., {\mathcal T}_m$, such that 
$\sum_{i:T_i\in {\mathcal T}_j}|T_i| \leq 3\frac{n}{m}$ for each 
${\mathcal T}_j$ \label{step:slbp2}
\State Let $U_1,...,U_m$ be any partition of $V'$ with each set of size at 
most $\sqrt{\frac{n}{\ln n}}$ \label{step:slbp1}
\State For each $j\in \{1,...,m\}$, 
let $V_j = U_j \cup \bigcup_{T_i\in{\mathcal T}_j}T_i$
\State \Return $\{V_1, ..., V_m\}$
\end{algorithmic}
\end{algorithm}

For the more complex Algorithm \ref{alg:slb}, 
we assume that $m > 2\sqrt{\frac{n}{\ln n}}$, because for 
lower values of $m$ the above simple algorithm gives the desired 
approximation. Also, the simple algorithm has better guarantee for 
all $n\leq e^{16}$, so when analyzing Algorithm  \ref{alg:slb}, we can assume 
that $n$ is sufficiently large for certain inequalities to hold, such as 
$\ln^{3} n < {n}$.  
The algorithm finds small disjoint sets of elements that have 
low ratio of  function value to size. Once a 
sufficient number of elements is grouped into such low-ratio sets, these 
sets are combined to form $m$ final sets of the partition, while adding a few 
remaining elements. These final sets have roughly $n/m$ elements each, so 
using submodularity and the low ratio property, we can bound the function 
value for each set in the partition.   

First we describe how some of the steps of algorithm work. The loop 
condition $|V'|>m\sqrt{\frac{n}{\ln n}}$ and our assumptions 
 $m > 2\sqrt{\frac{n}{\ln n}}$ and $\ln^{3} n < {n}$ imply that the probability 
$\frac{n}{m|V'|}$ (used on line \ref{line:prob}) is less than one.
The partition on line \ref{step:slbp1} can 
be found because at this point, the size of $V'$ is at most 
$m\sqrt{\frac{n}{\ln n}}$. For the partitioning done on line 
\ref{step:slbp2}, we note that since each $T_i$ is a subset of a sample set 
$S$ with $|S|\leq 2n/m$, it holds that $|T_i|\leq 2n/m$. Also, the total 
number of elements contained in all sets $T_i$ is at most $n$ (since they are 
disjoint). So  a simple greedy procedure that adds the sets $T_i$ to 
${\mathcal T}_j$ in arbitrary order, until the total number of elements is at 
least $n/m$, will produce at most $m$ groups, each with at most $3n/m$ 
elements. 

\begin{theorem}
\label{thm:slbub}
If given a feasible instance of the SLB problem, Algorithm \ref{alg:slb} 
outputs a solution of value at most $4\sqrt{\frac{n}{\ln n}} \cdot B$ with 
probability at least $p$. 
\end{theorem}
\begin{proof}
By monotonicity of $f$, the algorithm exits on line \ref{line:check} only if 
the instance is infeasible. Assume that the instance is feasible and 
let $\{U_1^*,\dots,U_m^*\}$ denote a solution with 
$\max_j f(U_j^*) < B$.
We consider one iteration of the while loop and show that with 
probability at least $\frac{c}{2n^2}$ it finds a set $T\subseteq S$ satisfying 
$f(T)-\alpha\cdot |T| < 0$. Then the probability that the size of $V'$ is 
reduced to $m\sqrt{\frac{n}{\ln n}}$ after 
$\frac{2n^3}{c} \ln(\frac{n}{1-p})$ iterations is at least $p$.

Assume, without loss of generality, that $U_1^*$ is the set that maximizes 
$|U^*_j \cap V'|$ for this iteration of the loop. If we let $n'=|V'|$, then 
$|U_1^* \cap V'|\geq \lceil n'/m \rceil$. 
Suppose  the sample $S$ found by the algorithm has 
size at most $2n/m$, and let $t=|U^*_1 \cap S|$ denote the size of the overlap 
of $S$ and $U_1^*$.  By monotonicity of $f$, we know that 
$f(U^*_1 \cap S)\leq f(U^*_1) < B$.  Since the 
algorithm finds a set $T\subseteq S$ minimizing the expression 
$f(T)-\alpha |T|$, we know that the value of this expression for $T$ is 
at most that for $U^*_1 \cap S$:
$$f(T)-\alpha |T| ~\leq ~ 
f(U^*_1 \cap S) - \alpha |U^*_1 \cap S|~<~ 
B - \frac{Bmt}{\sqrt{n \ln n}}.$$
In order to have $f(T)-\alpha |T| <0$, we need 
$t\geq \frac{\sqrt{n \ln n}}{m}$.  Next we show that the event that both 
$t\geq \frac{\sqrt{n \ln n}}{m}$ and $|S|\leq 2n/m$ happens with probability 
at least $\frac{c}{2n^2}$.

Let $x=\frac{\sqrt{n \ln n}}{m}$. 
To bound the probability that $t\geq x$, we focus on an arbitrary fixed subset of 
$U^*_1 \cap V'$ of size $\lceil n'/m \rceil$ (which is possible because 
$|U_1^* \cap V'|\geq \lceil n'/m \rceil$), and compute the probability that exactly $\lceil x \rceil$ elements from this subset make it into the sample $S$. In particular, this is the probability that sampling $\lceil n'/m \rceil$ items independently, with probability ${n}/{mn'}$  each, produces a sample of 
size $\lceil x \rceil$. We note that $x \in (1,n'/m)$, so $\lceil x \rceil$ is a valid sample size. These bounds follow because  inside the while loop, 
$m<n' \sqrt{{\ln n}/{n}}\leq \sqrt{n \ln n}$, so $x>1$. Also, 
$n'/m> \sqrt{{n}/{\ln n}}>{\ln n}> \sqrt{n\ln n}/m$ by the loop condition 
and our assumptions on $n$ and $m$, so $x< n'/m$. 
Let $\gamma,\delta \in [1,2)$ be such that $\gamma \cdot n'/m = \lceil n'/m \rceil$ and $\delta \cdot x = \lceil x \rceil$. 
We use an approximation derived from Stirling's formula as in the proof of Theorem \ref{thm:gen}. 

\begin{eqnarray}
\nonumber
\Pr[t=\lceil x \rceil] &\geq& {\gamma n'/m \choose \delta x} \cdot 
\left(\frac{n}{mn'}\right)^{\delta x} \cdot 
\left(1-\frac{n}{mn'}\right)^{\frac{\gamma n'}{m}-\delta x} \\
\nonumber
&\geq& \frac{c}{\sqrt{n}} \cdot 
\frac{\left(\frac{\gamma n'}{m}\right)^{\frac{\gamma n'}{m}} \cdot 
\left(\frac{n}{mn'}\right)^{\delta x} \cdot 
\left(1-\frac{n}{mn'}\right)^{\frac{\gamma n'}{m}-\delta x}}
{\left(\frac{\gamma n'}{m} \frac{\delta \sqrt{n \ln n}}{\gamma n'}\right)^{\delta x} \cdot 
\left(\frac{\gamma n'}{m} \left(1- \frac{\delta \sqrt{n \ln n}}{\gamma n'}\right)\right)^{\frac{\gamma n'}{m}-\delta x}}\\
\label{eqn:sl}
&=&  \frac{c}{\sqrt{n}} \cdot 
\left(\frac{n}{mn'}\frac{\gamma n'}{\delta \sqrt{n \ln n}}\right)^{\delta x} 
\left(\frac{1-\frac{n}{mn'}}{1-\frac{\delta \sqrt{n \ln n}}{\gamma n'}}\right)^{\frac{\gamma n'}{m}-\delta x}\\
\nonumber
&\geq&  \frac{c}{\sqrt{n}} \cdot 
\left(\frac{\gamma}{\delta m} \sqrt{\frac{n}{\ln n}}\right)^{\frac{\delta \sqrt{n \ln n}}{m}},
\end{eqnarray}
where the last inequality comes from observing that our assumption of 
$m> 2\sqrt{\frac{n}{\ln n}}$, together with $\gamma/\delta < 2$, imply that the last term on line (\ref{eqn:sl}) is greater than  1.

If we take a derivative of this bound with respect to $m$, which is 
$$\frac{\partial}{\partial m} \left(
 \frac{c}{\sqrt{n}} \cdot 
\left(\frac{\gamma}{\delta m} \sqrt{\frac{n}{\ln n}}\right)^{\frac{\delta \sqrt{n \ln n}}{m}} \right) =~ 
- \frac{c\,\delta \sqrt{\ln n}}{m^2} \cdot 
\left(\frac{\gamma}{\delta m} \sqrt{\frac{n}{\ln n}}\right) 
^{\frac{\delta \sqrt{n \ln n}}{m}} \cdot \left[\ln \left(\frac{\gamma}{\delta m} \sqrt{\frac{n}{\ln n}}\right)+1\right],
$$
and set it to zero, we find that the bound is minimized when 
$m=\frac{e\, \gamma}{\delta} \sqrt{\frac{n}{\ln n}}$. Substituting this value, 
$$\Pr[t=\lceil x\rceil] ~\geq~ c \cdot n^{-\frac{\delta^2}{e\,\gamma}-\frac{1}{2}} ~\geq~ 
c\cdot n^{-\frac{4}{e}-\frac{1}{2}} ~\geq~ c\cdot n^{-2}.$$

To bound the second probability, that $|S|\leq 2n/m$, we note that 
${\rm E}\left[|S|\right] = n/m$ and use 
Chernoff bound as well as the loop condition that implies
$m<n'\sqrt{\frac{\ln n}{n}}\leq \sqrt{n \ln n}$.  
$$\Pr\left[|S|>2\frac{n}{m}\right] ~<~ \left(\frac{e}{4}\right)^{\frac{n}{m}} ~\leq ~  
\left(\frac{e}{4}\right)^{\sqrt{\frac{n}{\ln n}}}$$
If  $n$ is sufficiently large that
$\left(\frac{e}{4}\right)^{\sqrt{\frac{n}{\ln n}}}\leq \frac{c}{2n^2}$, we can 
use the union bound to get
$$\Pr\left[t\geq x ~~{\rm and} ~~ |S|\leq 2\frac{n}{m}\right] ~\geq~ 
\frac{c}{n^2} - \frac{c}{2n^2} ~=~ \frac{c}{2n^2}.$$

This establishes that on feasible instances, the algorithm successfully 
terminates with probability at least $p$. Let us now consider the function 
value on any of the sets $V_j$ output by the algorithm.  By submodularity, 
$$f(V_j) ~\leq~ \sum_{v\in U_j}f(\{v\}) + \sum_{T_i\in {\mathcal T}_j} f(T_i).$$
For each $T_i$ we know that $f(T_i) < \alpha\cdot{|T_i|}$, and by the check 
performed on line \ref{line:check}, we have   
$f(\{v\})< B$ for each $v\in V$. Using this and the bounds on set sizes, 
$$f(V_j) ~\leq ~ 
 B \sqrt{\frac{n}{\ln n}} + \alpha \sum_{T_i\in {\mathcal T}_j} |T_i| ~\leq~
B \sqrt{\frac{n}{\ln n}} + \alpha \cdot \frac{3n}{m} =
B\cdot \left(\sqrt{\frac{n}{\ln n}}+\frac{m}{\sqrt{n \ln n}}\frac{3n}{m}\right) =
4\sqrt{\frac{n}{\ln n}} \cdot B. $$ 
\end{proof}


\section{Approximating submodular functions everywhere}

We present a lower bound for the problem of approximating submodular functions everywhere, 
which holds even for the special case of monotone functions. We use the same  
functions (\ref{eq:learnlb}) as for the SML lower bound in Section \ref{sec:sk}.

\begin{theorem}
\label{thm:learnlb}
Any algorithm that makes a polynomial number of oracle queries cannot 
approximate monotone submodular functions to a factor 
$o\left(\sqrt{\frac{n}{\ln n}}\right)$.
\end{theorem}
\begin{proof}
Suppose that there is a  $\gamma$-approximation algorithm 
for the problem, with $\gamma= o\left(\sqrt{\frac{n}{\ln n}}\right)$, 
which makes a polynomial number of oracle queries. 
Let $x=\sqrt{n}/\delta\gamma$, which satisfies $x^2= \omega(\ln n)$.
By Lemma \ref{lem:disting3}, with high probability this algorithm  produces 
the same output (say $\hat{f}$) if given as input either $f_3$ or $f_4$. Thus, by the algorithm's guarantee,  $\hat{f}$
 is simultaneously a $\gamma$-approximation for both $f_3$ and $f_4$.  
For the set $R$ used in $f_4$, this guarantee implies that 
$f_3(R)\leq \gamma\hat{f}(R) \leq \gamma f_4(R)$. Since $f_3(R)=\alpha$ and 
$f_4(R)=\beta$, we have that 
$\gamma \geq {\alpha}/{\beta} = {\sqrt{n}}/{x} = 2\gamma$, which is a 
contradiction. 
\end{proof}

\subsection{Approximating monotone two-partition submodular functions}
Recall that a 2P function is one for which there is a set $R\subseteq V$ such 
that the value of $f(S)$ depends only on 
$|S\cap R|$ and $|S\cap \bar{R}|$.  Our algorithm for approximating monotone 2P 
functions everywhere (Algorithm~\ref{alg:learn}) uses the following observation.

\begin{lemma}\label{lem:diffn}
Given two sets  $S$ and $T$ such that $|S|=|T|$, but $f(S)\neq f(T)$, a 2P 
function can be found exactly using a polynomial  number of oracle queries.
\end{lemma}
\begin{proof}
 This is done by inferring what the set $R$ is. Using $S$ and $T$, we 
find two sets which differ by exactly one element and have different 
function values. Fix an ordering of the elements of $S$, 
$\{s_1, ..., s_k\}$, and an ordering of elements of $T$, $\{t_1, ..., t_k\}$, 
such that the elements of $S\cap T$ appear last in both orderings, 
and in the same sequence.  
Let $S_0 = S$, and $S_i$ be the set $S$ with 
the first $i$ elements replaced by the first $i$ elements of $T$: 
$S_i = \{t_1,  ..., t_i,  s_{i+1}, ..., s_k\}$. Evaluate $f$ on each of 
the sets $S_i$ in order, until the first time that $f(S_{i-1})\neq f(S_{i})$. 
Such an $i$ must exist since $S_{k}=T$, and by assumption $f(T)\neq f(S)$.  
Let $U=\{t_1, ..., t_{i-1}, s_{i+1}, ..., s_{k}\}$, so that 
$S_{i-1}=U \cup \{s_i\}$ and $S_i = U \cup \{t_i\}$. 

The fact that 
$f(U\cup \{s_i\})\neq f(U\cup \{t_i\})$ tells us that either $s_i\in R$ and 
$t_i\notin R$, or  
vice versa.  Without loss of generality, we assume the former 
(since  the names of $R$ and $\bar{R}$ can be interchanged). 
Now all elements in $V\setminus U$ can be classified as belonging or not 
belonging to $R$.  In particular, if for some element $j\in \bar{U}$, 
$f(U\cup \{j\}) = f(U\cup \{s_i\})$, then $j\in R$; otherwise  
$f(U\cup \{j\}) = f(U\cup \{t_i\})$, and $j\notin R$. To test an element 
$u\in U$, evaluate $f(U - \{u\} + \{s_i, t_i\})$. This is the set $S_{i-1}$ 
with element $u$ replaced by $t_i$.  If $u\in \bar{R}$, then replacing one 
element from $\bar{R}$ by another will have no effect on the function value, 
and it will be equal to $f(S_{i-1})$. If $u\in R$, the we have replaced an 
element from $R$ by an element from $\bar{R}$, and we know that this changes 
the function value to $f(S_i)$. So all elements of $V$ can be tested for their 
membership in $R$, and then all function values can be obtained by querying 
sets $W$ with all possible values of $|W\cap R|$ and $|W\cap \bar{R}|$.
\end{proof}

\begin{algorithm}[ht] 
  \caption{Approximating a monotone 2P function everywhere. Input: $V,f,p$}\label{alg:learn}
\begin{algorithmic}[1]
\State
Query values of $f(\emptyset)$, $f(V)$, and $f(\{j\})$  for each $j\in V$
\State
For each $i\in \{2,..., n-1\}$, independently generate $n^{10} \ln \left(\frac{4n}{1-p}\right)$ random sets by including each element 
of $V$ into each set with probability $\frac{i}{n}$. Query the function value 
for each of these  sets.
\State
If the previous two steps produce any two sets $S_1$ and $S_2$ with $|S_1|=|S_2|$ and $f(S_1)\neq f(S_2)$, then find the function exactly, as described in Lemma \ref{lem:diffn}.
\State
Else, let $j\in V$ be an arbitrary element, and output 
$\hat{f}(S) ~=~ 
\begin{cases}
f(\emptyset) & \mbox{if} ~~ S = \emptyset \\
f(\{j\}) & \mbox{if} ~~ 1\leq |S| \leq 2\sqrt{n}\\
\frac{f(V)}{2\sqrt{n}} & \mbox{if} ~~ |S|> 2\sqrt{n}
\end{cases}
$
\end{algorithmic}
\end{algorithm}

\begin{theorem}
With probability at least $p$, the function $\hat{f}$ returned by  Algorithm 
\ref{alg:learn} satisfies \linebreak 
$\hat{f}(S) \leq f(S) \leq 2\sqrt{n} \cdot \hat{f}(S)$ for  
all sets $S\subseteq V$.
\end{theorem}
\begin{proof}
If the algorithm finds two sets $S_1$ and $S_2$ such that $|S_1|=|S_2|$ and 
$f(S_1)\neq f(S_2)$ during the sampling stage (steps 1 and 2), then the 
correctness of the output is implied by Lemma \ref{lem:diffn}.  If it does not 
find such sets, then it outputs the function $\hat{f}$ shown in step 4. It 
obviously satisfies the inequality for the case that $S=\emptyset$. For the 
case that $1\leq |S|\leq 2\sqrt{n}$, we observe that if the algorithm reaches step 4, 
it must be that the value of $f$ is identical for all singleton sets,  
i.e. $f(\{j\})=f(\{j'\})$ for all $j, j'\in V$. Now, $f(S)\geq f(\{j\}) = \hat{f}(S)$ by 
monotonicity. Also, by submodularity, 
$f(S)\leq \sum_{j\in S} f(\{j\}) = |S|\cdot \hat{f}(S)\leq 2\sqrt{n}\cdot \hat{f}(S)$, 
establishing the correctness for the case that $|S|\leq 2\sqrt{n}$. For the 
last case, $|S|> 2\sqrt{n}$, the inequality 
$f(S)\leq f(V) = 2\sqrt{n}\cdot \hat{f}(S)$ follows by monotonicity. For the 
other one, $\hat{f}(S)\leq f(S)$, we need an additional nontrivial lemma.

Since the 2P function $f(S)$ depends only on two values, $|S\cap R|$ and 
$|S\cap \bar{R}|$, let us denote by $f(k,l)$ the value of the function $f$ on 
a set $S$ with $|S\cap R|=k$ and $|S\cap \bar{R}|=l$. We say that such a 
set $S$ corresponds to the pair $(k,l)$.
We assume that $0<|R|<n$, because if $|R|=0$ or $|R|=n$, then $f(S)$ is a 
function that depends only on $|S|$, and it equally well can be represented as 
a 2P function with any other set $\hat{R}$.  Furthermore, we 
assume without loss of generality that $|R|\leq |\bar{R}|$ (otherwise 
interchange $R$ and $\bar{R}$), and let $K=|R|$ and $L=|\bar{R}|$ (which 
are not known to the algorithm).

\begin{lemma} \label{lem:walk}
For any $k$ and any $l$, $f(k,0)\geq \frac{k}{2n}  f(V)$ and 
$f(0,l)\geq \frac{l}{2n}  f(V)$.
\end{lemma}

Using this lemma to finish the proof,
let $k=|S\cap R|$ and 
$l= |S\cap \bar{R}|$. We observe that by monotonicity, 
$f(S)\geq f(k,0)$ and $f(S) \geq f(0,l)$. Moreover, 
since $|S|=k+l \geq 2\sqrt{n}$, we have $\max(k,l)\geq {\sqrt{n}}$. 
So by Lemma~\ref{lem:walk}, 
$f(S)\geq \frac{\max(k,l)}{2n} f(V)\geq \frac{f(V)}{2\sqrt{n}} = \hat{f}(S)$. 
\end{proof} 

\medskip 

The proof of Lemma~\ref{lem:walk} is involved, and we first sketch the main ideas.  We call a pair $(k,l)$ \emph{balanced} if 
$k/l$ is close to $K/L$. Then, with significant probability, the algorithm 
samples  sets corresponding to all balanced pairs. 
Since the algorithm checks for sets of the same size with  
different function values, we can assume that if it
proceeds to step 4, then for  sets $S$  corresponding to balanced pairs, 
$f(S)$ is a function $F$ that depends only on $|S|$. We use submodularity to 
show that  $F$ is concave.   Then we decompose $f(k,0)$ as 
$\sum_{i=1}^k [f(i,0)-f(i-1,0)]$ and lower-bound each term in this sum 
separately by comparing it to an increment
$f(i,j)-f(i-1,j)$ for some $j$ with $(i,j)$ balanced.
Then, using  concavity of $F$, we lower-bound their sum.

To prove Lemma~\ref{lem:walk}, we use a definition and several preliminary lemmas.

\begin{wrapfigure}{R}{5cm}
  \centerline{ \includegraphics[draft=false,scale=0.5]{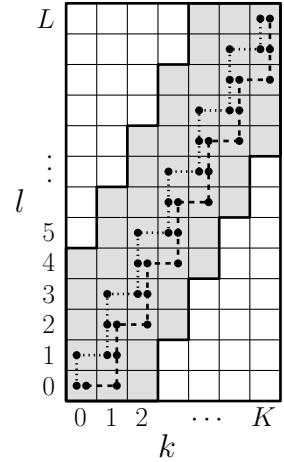} }
  \caption{In the table of pairs $(k,l)$, the shaded cells 
 correspond to  balanced pairs, and the $K$-biased (dashed) and $L$-biased 
 (dotted) walks are shown.}
\label{fig:walk}
\end{wrapfigure}

\begin{definition}
A pair of integers $(k,l)$ with $k\leq K$ and $l\leq L$ is said to be \emph{balanced} if 
it satisfies 
\begin{equation} \label{eqn:kl}
l\cdot \frac{K}{L}-2 ~\leq~ k ~\leq~ l\cdot \frac{K}{L}+2.
\end{equation}
\end{definition}

Intuitively, 
in a set corresponding to a balanced pair, the 
numbers of elements from $R$ and $\bar{R}$ are proportional to the sizes of the two sets (see Figure \ref{fig:walk}).

\begin{lemma} \label{lem:adpr}
Suppose that $m\leq n$ elements are selected independently with probability 
$q\in [\frac{1}{n}, \frac{n-1}{n}]$ 
each, and let $X$ denote the total number of selected elements.  Then for any 
integer $x\in [0,m-1]$, 
$$\frac{1}{n^2}~~ \leq ~~\frac{\Pr[X=x+1]}{\Pr[X=x]} ~~\leq~~ n^2.$$
\end{lemma}
\begin{proof}
$$\frac{\Pr[X=x+1]}{\Pr[X=x]} ~=~ \frac{{m \choose x+1} ~q^{x+1}~ (1-q)^{m-x-1} }
{{m \choose x}~ q^{x} ~(1-q)^{m-x}} ~=~ 
\frac{(m-x)~q}{(x+1)(1-q)},$$
with the minimum value of $1/m(n-1)\geq 1/n^2$ achieved 
 at $x=m-1$ and $q=\frac{1}{n}$,  and the  maximum value of $m(n-1)\leq n^2$ 
achieved at $x=0$ and $q=\frac{n-1}{n}$.
\end{proof}

\begin{lemma} \label{lem:diag}
If Algorithm \ref{alg:learn} reaches step 4, then with probability at least $p$,  
for all balanced $(k_1,l_1)$ and $(k_2,l_2)$ such that $k_1+l_1=k_2+l_2$, 
it holds that $f(k_1, l_1) = f(k_2, l_2)$. In other words, for all balanced 
pairs $(k,l)$, the value of  $f(k,l)$ depends only on $k+l$.
\end{lemma}
\begin{proof}
The lemma follows if we show that with probability at least $p$, for each 
balanced $(k,l)$ with $k+l<n$, the algorithm samples at least one set $S$ corresponding to $(k,l)$.
This is because the algorithm verifies that the function value for the sets 
that it samples depends only on the set size. 

So consider a specific balanced pair $(k,l)$ and one random 
set $S$ generated by the 
iteration $i=k+l$ of step 2 of the algorithm. The probability of 
sampling each element in this iteration is $q=\frac{i}{n}=\frac{k+l}{K+L}$. Using 
(\ref{eqn:kl}) and its equivalent $(k-2)L/K \leq l \leq (k+2)L/K$,
 we see that this probability satisfies the following:
$$\frac{k}{K}-\frac{2L}{Kn} ~\leq ~ q ~\leq ~ \frac{k}{K}+\frac{2L}{Kn}~~~ {\rm and}~~~
\frac{l}{L}-\frac{2}{n} ~\leq ~ q ~\leq ~\frac{l}{L}+\frac{2}{n}.$$
So the expected value of $|S\cap R|$ is 
$qK\in [k-2L/n, k+2L/n] \subseteq [k-2,k+2]$. Similarly, the expected value of 
$|S\cap \bar{R}|$ is $qL \in [l-2, l+2]$.  
Let $\mu_k$ be the most likely number of sampled elements when independently sampling $K$ 
elements with probability $q$ each. Then $\mu_k$ is equal to either 
$\lfloor qK \rfloor$ or $\lceil qK \rceil$. From above considerations and 
because $k$ is an integer, we have that $\mu_k\in [k-2,k+2]$. 
Now, since $\mu_k$ is the most likely value, we know that 
$\Pr[|S\cap R|=\mu_k]\geq 1/(K+1) \geq 1/n$. By Lemma \ref{lem:adpr} (with $m=K$), 
$$\Pr[|S\cap R|=k] ~\geq~ \Pr[|S\cap R|=\mu_k] \cdot n^{-2\cdot|k-\mu_k|} ~\geq~ n^{-5}.$$
We similarly define $\mu_l$, observe that $\mu_l\in [l-2,l+2]$, and conclude 
that $\Pr[|S\cap \bar{R}|=l]\geq n^{-5}$. Since the two events are 
independent, the probability that both of them occur, and thus that $S$ 
corresponds to $(k,l)$, is at least $n^{-10}$. 

We observe that for any $i$, there are at most four balanced pairs $(k,l)$ such that 
$k+l=i$.  This is because if some pair $(k,l)$ 
satisfies (\ref{eqn:kl}), then the pair $(k-4, l+4)$ doesn't satisfy it:
$$k-4 ~\leq ~ \left(l\frac{K}{L} +2\right) -4 ~=~ l\frac{K}{L}-2 ~<~ (l+4)\frac{K}{L}-2.$$
So there is a total of at most $4n$ pairs $(k,l)$ for which we would like the 
algorithm to sample their corresponding sets. 
Since the number of trials for each value of $k+l$ is 
$n^{10} \ln \left(\frac{4n}{1-p}\right)$, the probability that 
a set corresponding to any particular pair $(k,l)$ is \emph{not} sampled is at most
$$ \left(1-n^{-10}\right)^{n^{10} \ln \left(\frac{4n}{1-p}\right)}
~\leq~ e^{-\ln \left(\frac{4n}{1-p}\right)} ~=~ \frac{1-p}{4n}.$$
Since there are at most $4n$ pairs of interest, by union bound we have that 
the probability that at least one of them remains unsampled is at most $(1-p)$. 
\end{proof}

\medskip

Suppose the condition in Lemma \ref{lem:diag} holds. Let us define a function 
$F(i)$ to be equal to $f(k,l)$ such that $k+l=i$ and $(k,l)$ is balanced. 
$F(i)$ is defined for all $i\in \{0,...,n\}$, since for any such $i$ there is 
at least one balanced pair $(k,l)$ with $k+l=i$.

\begin{lemma}
$F(i)$ is a non-decreasing concave function.
\end{lemma}
\begin{proof}
Let $\Delta(i) = F(i+1)-F(i)$. It suffices to show that the sequence of increments 
$\Delta(i)$ is non-negative and non-increasing.
For any $i$, we define a pair $(k_i,l_i) = 
\left(\left\lfloor \frac{iK}{n} \right\rfloor, \left\lceil \frac{iL}{n} 
\right\rceil \right)$. It can be verified that all pairs $(k_i, l_i)
$ as well as $(k_i+1, l_i)$ are balanced. Furthermore, $k_i+l_i=i$ 
 (and consequently  $k_i+1+l_i=i+1$), so that $f(k_i+1,l_i)-f(k_i,l_i)=\Delta(i)$. 
Also, both  $\{k_i\}$ and $\{l_i\}$ are non-decreasing sequences. The 
decreasing marginal values of the submodular function $f$ imply that 
$\Delta(i+1)= f(k_{i+1}+1,l_{i+1})-f(k_{i+1},l_{i+1}) \leq 
f(k_i+1,l_i)-f(k_i,l_i) = \Delta(i)$, 
showing that $\Delta(i)$'s are non-increasing. The monotonicity of $f$ implies that 
they are also non-negative.
\end{proof}

\medskip

We next define two sequences of pairs, $(k^K_i, l^K_i)$ and $(k^L_i, l^L_i)$, 
ranging from $i=0$ to $i=n$, 
which we call the \emph{$K$-biased} sequence (or walk) and the \emph{$L$-biased} 
sequence, respectively (see Figure \ref{fig:walk}). The properties of these 
two sequences will be used in  
the remainder of the proof. The definitions are inductive, with both sequences 
starting at $(0,0)$.
\begin{eqnarray*}
(k^K_{i+1}, ~l^K_{i+1}) &=& 
\begin{cases}
(k^K_i +1, ~l^K_i) & {\rm if }~~ k^K_i \leq l^K_i \cdot \frac{K}{L}\\
(k^K_i, ~l^K_i+1) & {\rm if }~~ k^K_i > l^K_i \cdot \frac{K}{L}
\end{cases}
\\
(k^L_{i+1}, ~l^L_{i+1}) &=& 
\begin{cases}
(k^L_i +1, ~l^L_i) & {\rm if }~~ k^L_i < l^L_i \cdot \frac{K}{L}\\
(k^L_i, ~l^L_i+1) & {\rm if }~~ k^L_i \geq l^L_i \cdot \frac{K}{L}
\end{cases}
\end{eqnarray*}
Let us call the change from $(k_i,l_i)$ to $(k_{i+1}, l_{i+1})$ in either of 
the two sequences a \emph{$K$-step} 
if the first component of the pair increases by one, and an \emph{$L$-step} if 
the second component increases.  The only difference between the two sequences 
is that when equality $k=l\cdot K/L$ holds, we take a $K$-step in the case of 
the $K$-biased sequence, and an $L$-step in the case of the $L$-biased 
sequence. 
For both sequences it holds that $k^K_i+l^K_i=k^L_i+l^L_i = i$,~ $k^K_i$ and 
$k^L_i$ range between $0$ and $K$, and $l^K_i$ and 
$l^L_i$ range between $0$ and $L$.

\begin{lemma}
All pairs in the $K$-biased and $L$-biased sequences are balanced.
\end{lemma}
\begin{proof}
The proof is by induction, and it is the same for both sequences, so we denote 
either sequence by $(k_i,l_i)$. The first pair $(0,0)$ is balanced. Now we assume 
that the pair $(k_i,l_i)$ is balanced, and would like to show that the pair 
$(k_{i+1},l_{i+1})$ is also balanced.  Suppose 
$(k_{i+1}, l_{i+1}) = (k_i+1, l_i)$. Then it must be that 
$k_i \leq l_i\cdot \frac{K}{L}$. Then
$$l_i\cdot \frac{K}{L} - 2 ~\leq~ k_i ~\leq~~ k_i +1 ~~\leq ~ l_i\cdot \frac{K}{L} +1.$$
If $(k_{i+1}, l_{i+1}) = (k_i, l_i+1)$, then it must be that $k_i \geq l_i\cdot \frac{K}{L}$. Then
$$(l_i+1)\cdot \frac{K}{L} -2 ~\leq~ l_i\cdot\frac{K}{L}~ \leq ~~k_i ~~\leq ~l_i\cdot\frac{K}{L}+2 ~\leq~ (l_i+1)\cdot \frac{K}{L} +2,$$
with the leftmost inequality following because $K/L\leq 1$.
\end{proof}

\begin{lemma} \label{lem:alt}
In the $K$-biased sequence, every $K$-step is followed by at most 
$\left\lceil \frac{L}{K} \right\rceil$ $L$-steps. In the $L$-biased sequence, every $L$-step is 
followed by at most one $K$-step.
\end{lemma}
\begin{proof}
Suppose that the $K$-biased sequence, after some point $(k,l)$, takes one $K$-step 
followed by $\left\lceil \frac{L}{K} \right\rceil$ $L$-steps, reaching the 
point $(k+1, l+\left\lceil \frac{L}{K} \right\rceil)$. 
Since the step after $(k,l)$ is 
a $K$-step, it must be that $k\leq lK/L$. So
$$\left(l+\left\lceil \frac{L}{K} \right\rceil \right)\cdot \frac{K}{L} ~\geq~
l\cdot \frac{K}{L} + 1 ~\geq~ k+1,$$
which means that the next step in the $K$-biased sequence will be a $K$-step.

Similarly, for the $L$-biased walk, suppose that from some point $(k,l)$, the 
sequence takes an $L$-step, followed by a $K$-step, reaching the point 
$(k+1,l+1)$. Then $k\geq lK/L$ implies that
$$(l+1)\cdot \frac{K}{L} ~=~ l\cdot\frac{K}{L}+\frac{K}{L} ~\leq~ 
l\cdot\frac{K}{L}+1 ~\leq~ k+1,$$
and thus the next step is an $L$-step.
\end{proof}

\medskip

\begin{proof}[Proof of Lemma \ref{lem:walk}.]
To lower-bound the value of $f(k,0)$, we consider the $K$-biased walk from 
$(0,0)$ to a point $(k,l')$ which is the last point before the $K$-step to 
$(k+1,\cdot)$.  We let $f(k,0)= F(0)+\sum_{j=1}^k \delta(j)$,  
where $\delta(j) = f(j,0)-f(j-1,0)$. For each $K$-step in the $K$-biased walk, 
where $k^K_{i-1}=j-1$ and $k^K_{i}=j$, let 
$\Delta^K(j) = f(k^K_i, l^K_i)-f(k^K_{i-1}, l^K_{i-1}) = 
f(j, l^K_i)-f(j-1, l^K_{i-1})$. By submodularity of $f$ it follows that 
$\Delta^K(j)\leq \delta(j)$.  

We claim that $\sum_{j=1}^k \Delta^K(j) 
\geq [f(k,l')-F(0)]/(1+ \left\lceil \frac{L}{K} \right\rceil )$.  
In other words, at least $1/(1+ \left\lceil \frac{L}{K} \right\rceil )$
fraction of the increase in $F(\cdot)$, as we proceed in the $K$-biased walk, 
is due to the $K$-steps. This follows from several observations. First, 
the $K$-biased walk starts with a $K$-step. Second, by Lemma \ref{lem:alt}, 
each $K$-step is followed by no more than 
$\left\lceil \frac{L}{K} \right\rceil$ $L$-steps. And third, $\Delta^K(j)$ is a 
decreasing sequence (by concavity of $F$). 

Further, by concavity of $F$, we have that 
$f(k,l') \geq \frac{k+l'}{n} F(n)$. By definition of $l'$, 
we have $l'\geq kL/K$. Also, $1+\left\lceil \frac{L}{K} \right\rceil \leq 2(L/K+1)$.
Putting everything together, we have
\begin{eqnarray*}
f(k,0) &=& F(0)+\sum_{j=1}^k \delta(j) ~\geq ~
F(0) + \sum_{j=1}^k\Delta^K(j) ~\geq~
F(0) + \frac{f(k,l')-F(0)}{1+ \left\lceil \frac{L}{K} \right\rceil } ~\geq~
\frac{f(k,l')}{1+ \left\lceil \frac{L}{K} \right\rceil}\\
&\geq & \frac{k+l'}{n} \frac{F(n)}{2 (L/K+1)} 
~\geq~ \frac{k(L/K+1)}{n} \frac{F(n)}{2 (L/K+1)} 
~=~ \frac{k}{2n}F(n)
\end{eqnarray*}

To bound $f(0,l)$, we consider the $L$-biased walk from $(0,0)$ to $(k',l)$ 
for some $k'$. Because of concavity of $F$, the $L$-steps in the walk account 
for at least half the increase in $f$, yielding 
$f(0,l)\geq \frac{1}{2} f(k',l)$. Also, 
$f(k',l)\geq \frac{k'+l}{n} F(n) \geq \frac{l}{n} F(n)$. So we get that 
$f(0,l)\geq \frac{l}{2n} F(n)$.
\end{proof}


\section{Acknowledgements}
We thank Mark Sandler for his help with some of the calculations and Satoru Iwata for useful discussions.


{\small
\bibliographystyle{plain}
\bibliography{../../bib/names,../../bib/conferences,../../bib/bibliography}

\begin{thebibliography}{10}

\bibitem{AroraHK04}
S.~Arora, E.~Hazan, and S.~Kale.
\newblock ${O}(\sqrt{\log n})$ approximation to sparsest cut in
  $\tilde{O}(n^2)$ time.
\newblock In {\em Proc. 45th IEEE Symp. on Foundations of Computer Science},
  pages 238--247, 2004.

\bibitem{arora:rao:vazirani:sqrt-logn}
S.~Arora, S.~Rao, and U.~Vazirani.
\newblock Expander flows, geometric embeddings and graph partitioning.
\newblock In {\em Proc. 36th ACM Symp. on Theory of Computing}, 2004.

\bibitem{calinescu:submod}
G.~Calinescu, C.~Chekuri, M.~Pal, and J.~Vondrak.
\newblock Maximizing a submodular set function subject to a matroid constraint.
\newblock {\em SIAM J. Comput.}
\newblock To appear in STOC 2008 special issue.

\bibitem{calinescu:zelikovsky}
G.~Calinescu and A.~Zelikovsky.
\newblock The polymatroid {S}teiner problems.
\newblock {\em J. Comb. Optim.}, 9(3):281--294, 2005.

\bibitem{chekuri:pal:recursive}
C.~Chekuri and M.~Pal.
\newblock A recursive greedy algorithm for walks in directed graphs.
\newblock In {\em Proc. 46th IEEE Symp. on Foundations of Computer Science},
  pages 245--253, 2005.

\bibitem{clrs}
T.~Cormen, C.~Leiserson, R.~Rivest, and C.~Stein.
\newblock {\em Introduction to Algorithms}.
\newblock MIT Press, second edition, 2001.

\bibitem{Cunningham85}
W.H. Cunningham.
\newblock Minimum cuts, modular functions, and matroid polyhedra.
\newblock {\em Networks}, 15:205--215, 1985.

\bibitem{feige-maxsubmod}
U.~Feige, V.~Mirrokni, and J.~Vondrak.
\newblock Maximizing non-monotone submodular functions.
\newblock In {\em Proc. 48th IEEE Symp. on Foundations of Computer Science},
  2007.

\bibitem{FleischerIwata03}
L.~Fleischer and S.~Iwata.
\newblock A push-relabel framework for submodular function minimization and
  applications to parametric optimization.
\newblock {\em Discrete Appl. Math.}, 131(2):311--322, 2003.

\bibitem{fujishige:polymatroid}
S.~Fujishige.
\newblock Polymatroid dependence structure of a set of random variables.
\newblock {\em Info. and Control}, 39:55--72, 1978.

\bibitem{gens:levner}
G.V. Gens and E.V. Levner.
\newblock Computational complexity of approximation algorithms for
  combinatorial problems.
\newblock In {\em Proc. 8th Intl. Symp. on Math. Foundations of Comput. Sci.}
  Lecture Notes in Comput. Sci. 74, Springer-Verlag, 1979.

\bibitem{goel:karande}
G.~Goel, C.~Karande, P.~Tripathi, and L.~Wang.
\newblock Approximability of combinatorial problems with multi-agent submodular
  cost functions.
\newblock In {\em Proc. 50th IEEE Symp. on Foundations of Computer Science},
  2009.

\bibitem{goemans:learning}
M.~Goemans, N.~Harvey, S.~Iwata, and V.~Mirrokni.
\newblock Approximating submodular functions everywhere.
\newblock In {\em Proc. 20th ACM Symp. on Discrete Algorithms}, 2009.

\bibitem{goemans:harvey:kleinberg:mirrokni}
M.~Goemans, N.~Harvey, R.~Kleinberg, and V.~Mirrokni.
\newblock Unpublished manuscript.

\bibitem{GLS81}
M.~Gr{\"{o}}tschel, L.~Lov{\'a}sz, and A.~Schrijver.
\newblock The ellipsoid method and its consequences in combinatorial
  optimization.
\newblock {\em Combinatorica}, 1:169--197, 1981.

\bibitem{GLS88}
M.~Gr{\"o}tschel, L.~Lov{\'a}sz, and A.~Schrijver.
\newblock {\em Geometric Algorithms and Combinatorial Optimization}.
\newblock Springer-Verlag, 1988.

\bibitem{hayrapetyan:kempe:pal:svitkina}
A.~Hayrapetyan, D.~Kempe, M.~Pal, and Z.~Svitkina.
\newblock Unbalanced graph cuts.
\newblock In {\em Proc. 13th European Symposium on Algorithms}, 2005.

\bibitem{hayrapetyan:swamy:tardos}
A.~Hayrapetyan, C.~Swamy, and E.~Tardos.
\newblock Network design for information networks.
\newblock In {\em Proc. 16th ACM Symp. on Discrete Algorithms}, pages 933--942,
  2005.

\bibitem{HochbaumShmoys87}
D.~S. Hochbaum and D.~B. Shmoys.
\newblock Using dual approximation algorithms for scheduling problems:
  theoretical and practical results.
\newblock {\em J. ACM}, 34:144--162, 1987.

\bibitem{Iwata03}
S.~Iwata.
\newblock A faster scaling algorithm for minimizing submodular functions.
\newblock {\em SIAM J. Comput.}, 32:833--840, 2003.

\bibitem{Iwata08}
S.~Iwata.
\newblock Submodular function minimization.
\newblock {\em Math. Programming}, 112:45--64, 2008.

\bibitem{iwata:fleischer:fujishige}
S.~Iwata, L.~Fleischer, and S.~Fujishige.
\newblock A combinatorial strongly polynomial algorithm for minimizing
  submodular functions.
\newblock {\em J. ACM}, 48(4):761--777, 2001.

\bibitem{iwata:nagano}
S.~Iwata and K.~Nagano.
\newblock Submodular function minimization under covering constraints.
\newblock In {\em Proc. 50th IEEE Symp. on Foundations of Computer Science},
  2009.

\bibitem{iwata:orlin}
S.~Iwata and J.~B. Orlin.
\newblock A simple combinatorial algorithm for submodular function
  minimization.
\newblock In {\em Proc. 20th ACM Symp. on Discrete Algorithms}, 2009.

\bibitem{kulik}
A.~Kulik, H.~Shachnai, and T.~Tamir.
\newblock Maximizing submodular set functions subject to multiple linear
  constraints.
\newblock In {\em Proc. 20th ACM Symp. on Discrete Algorithms}, 2009.

\bibitem{lee:nonmon}
J.~Lee, V.~Mirrokni, V.~Nagarajan, and M.~Sviridenko.
\newblock Non-monotone submodular maximization under matroid and knapsack
  constraints.
\newblock In {\em Proc. 41th ACM Symp. on Theory of Computing}, 2009.

\bibitem{lee:sviridenko:vondrak}
J.~Lee, M.~Sviridenko, and J.~Vondrak.
\newblock Submodular maximization over multiple matroids via generalized
  exchange properties.
\newblock In {\em Proc. 12th APPROX}, 2009.

\bibitem{leighton:rao}
F.T. Leighton and S.~Rao.
\newblock Multicommodity max-flow min-cut theorems and their use in designing
  approximation algorithms.
\newblock {\em Journal of the ACM}, 46, 1999.

\bibitem{LenstraST90}
J.~K. Lenstra, D.~B. Shmoys, and E.~Tardos.
\newblock Approximation algorithms for scheduling unrelated parallel machines.
\newblock {\em Math. Programming}, 46:259--271, 1990.

\bibitem{motwani:raghavan}
R.~Motwani and P.~Raghavan.
\newblock {\em Randomized Algorithms}.
\newblock Cambridge University Press, 1990.

\bibitem{nemhauser:wolsey:fisher}
G.~Nemhauser, L.~Wolsey, and M.~Fisher.
\newblock An analysis of the approximations for maximizing submodular set
  functions.
\newblock {\em Math. Program.}, 14:265--294, 1978.

\bibitem{Orlin07}
J.~B. Orlin.
\newblock A faster strongly polynomial time algorithm for submodular function
  minimization.
\newblock {\em Math. Programming}.
\newblock To appear.

\bibitem{Queyranne98}
M.~Queyranne.
\newblock Minimizing symmetric submodular functions.
\newblock {\em Math. Programming}, 82:3--12, 1998.

\bibitem{racke:optimal}
H.~R\"acke.
\newblock Optimal hierarchical decompositions for congestion minimization in
  networks.
\newblock In {\em Proc. 40th ACM Symp. on Theory of Computing}, pages 255--263,
  2008.

\bibitem{schrijver:combinatorial}
A.~Schrijver.
\newblock A combinatorial algorithm minimizing submodular functions in strongly
  polynomial time.
\newblock {\em J. of Combinatorial Theory, Ser. B}, 80(2):346--355, 2000.

\bibitem{sviridenko:note}
M.~Sviridenko.
\newblock A note on maximizing a submodular set function subject to a knapsack
  constraint.
\newblock {\em Oper. Res. Lett.}, 32(1):41--43, 2004.

\bibitem{svitkina:tardos:multiway-cuts}
Z.~Svitkina and E.~Tardos.
\newblock Min-max multiway cut.
\newblock In {\em Proc. 7th APPROX}, pages 207--218, 2004.

\bibitem{svitkina:tardos:facility}
Z.~Svitkina and E.~Tardos.
\newblock Facility location with hierarchical facility costs.
\newblock {\em ACM Transactions on Algorithms}, 6(2), 2010.

\bibitem{SwamySW07}
C.~Swamy, Y.~Sharma, and D.~Williamson.
\newblock Approximation algorithms for prize collecting steiner forest problems
  with submodular penalty functions.
\newblock In {\em Proc. 18th ACM Symp. on Discrete Algorithms}, 2007.

\bibitem{vondrak:symmetry}
J.~Vondrak.
\newblock Symmetry and approximability of submodular maximization problems.
\newblock In {\em Proc. 50th IEEE Symp. on Foundations of Computer Science},
  2009.

\bibitem{wolsey:analysis}
L.~A. Wolsey.
\newblock An analysis of the greedy algorithm for the submodular set covering
  problem.
\newblock {\em Combinatorica}, 2(4):385--393, 1982.

\bibitem{ZhaoNagamochiIbaraki}
L.~Zhao, H.~Nagamochi, and T.~Ibaraki.
\newblock Greedy splitting algorithms for approximating multiway partition
  problems.
\newblock {\em Math. Program.}, 102(1):167--183, 2005.

\end{thebibliography}
}


\end{document}